\documentclass[dvipsnames]{lmcs}
\pdfoutput=1

\usepackage{lastpage}
\lmcsdoi{18}{3}{30}
\lmcsheading{}{\pageref{LastPage}}{}{}%
{Nov.~02,~2021}{Sep.~08,~2022}{}

\usepackage[utf8]{inputenc}
\title{Distributed Asynchronous Games With Causal Memory are Undecidable}

\usepackage{amsmath}
\usepackage{amsthm}
\usepackage{amsfonts}
\usepackage{alltt}
\usepackage{hyperref}
\usepackage{latexsym}
\usepackage{amssymb}
\usepackage{amsmath}
\usepackage{amsthm}

\usepackage{pgf}
\usepackage{tikz}
\usetikzlibrary{arrows,automata,shapes,calc}

\newcommand{\dsp}{{\sc distributed control problem}}

\newcommand{\lcp}{{\sc bipartite coloring problem}}

\usepackage{alltt}

\newcommand{\ignore}[1]{}

\newcommand{\NN}{\mathbb{N}}

\newcommand{\mycolor}{BurntOrange}

\newcommand{\PP}{\mathbb{P}}

\renewcommand{\AA}{\mathcal{A}}

\newcommand{\bi}{\begin{itemize}}
\newcommand{\ei}{\end{itemize}}

\newcommand{\ind}{~\mathbb{I}~}

\newcommand{\pref}{\sqsubseteq}

\newcommand{\be}{\begin{equation}}
\newcommand{\ee}{\end{equation}}
\DeclareMathOperator{\dom}{dom}
\DeclareMathOperator{\view}{\partial}

\DeclareMathOperator{\last}{last}
\DeclareMathOperator{\nex}{next}
\DeclareMathOperator{\prev}{prev}

\DeclareMathOperator{\plays}{plays}
\DeclareMathOperator{\alphabet}{Alph}
\DeclareMathOperator{\state}{state}

\newcommand{\PCP}{{\sc Post Correspondence Problem}}

\newcommand{\yes}{{\tt SameLength}}
\newcommand{\no}{{\tt -}}
\newcommand{\terminate}{\ensuremath{\mathtt{END}}}
\newcommand{\increment}{\ensuremath{\mathtt{I}}}
\newcommand{\chk}{\ensuremath{\mathtt{CHECK}}}
\newcommand{\answer}{\ensuremath{\mathtt{ANSWER}}}
\newcommand{\sss}{{\tt End}}
\newcommand{\lose}{\ensuremath{\mathtt{LOSE}}}
\newcommand{\win}{\ensuremath{\mathtt{WIN}}}

\newcommand{\yess}{{\tt SameTile}}
\newcommand{\noo}{{\tt -}}

\definecolor{amber}{rgb}{1.0, 0.75, 0.0}

\author[H.~Gimbert]{Hugo Gimbert}	%required
\address{LaBRI, CNRS, Université de Bordeaux, France}	%required
\email{hugo.gimbert@cnrs.fr}

\begin{document}

\begin{abstract}
We show the undecidability of the distributed control problem
when the plant is an asynchronous automaton,
the controllers use causal memory and
the goal of the controllers is to put each process
in a local accepting state.
 \end{abstract}

\maketitle

 \section*{Introduction}
  The decidability of the distributed version of the Ramadge and Wonham control problem~\cite{ramadge1989control},
where both the plant and the controllers are modeled as asynchronous automata~\cite{zautomata,thebook}
and the controllers have causal memory
has been an open problem for some time~\cite{gastin,alook,mumu}.

In this setting a controllable plant is distributed on several finite-state processes
which interact asynchronously using shared actions.
On every process, the local controller can choose to block some of the actions,
called \emph{controllable} actions, but it cannot block the \emph{uncontrollable} actions from the environment.
The choices of the local controllers are based on two sources of information.
First, the local controller monitors the sequence of states and actions of the local process.
This information is called the \emph{local view} of the controller.
Second, when a shared action is played by several processes
then all the controllers of these processes can exchange as much information as they want.
In particular together they can compute their mutual view of the global execution:
their \emph{causal past}.

A  controller is correct if it guarantees that every
possible execution of the plant
satisfies some specification. The controller synthesis problem is a decision problem which, given a plant as input,
asks whether the system admits a correct controller.
In case such a controller exists, the algorithm should compute one as well.

The difficulty of controller synthesis depends on several factors, including the architecture of the information flow between processes (pipeline, ring, ...), % chktex 26 chktex 11
 the information available to the controllers,
and the specification.

In early work on distributed controller synthesis,
for example in the setting of~\cite{pneuli1990distributed}, the only source of information available to the controllers is their local view,
and the transmission of information between controllers is strictly limited to reading and writing shared variables with fixed domains.
In this setting, distributed synthesis is not decidable in general, but decidable for architectures without information forks~\cite{finkbeiner2005uniform}.

The setting where controllers are allowed to freely communicate with each other upon synchronisation,
and can rely on their full causal past
to select the controllable actions
appeared first in~\cite{gastin}.
In this setting, the class of decidable architectures goes far beyond architectures without information forks,
for example all series-parallel games are decidable~\cite{gastin}.

\medskip

We adopt a modern terminology and call the plant a \emph{distributed game} and the controllers are \emph{distributed strategies} in this game.
A distributed strategy is a function that maps the causal past of processes
to a subset of controllable actions.
In the present paper we focus on the \emph{termination condition},
which is satisfied when each process is guaranteed to terminate its computation in finite time,  in a final state.
A distributed strategy is winning if it guarantees the termination condition,
whatever uncontrollable actions are chosen by the environment.

We are interested in the following algorithmic problem:
\smallskip

\noindent \dsp: given a distributed game decide
whether there exists a distributed winning strategy.

\smallskip

\paragraph*{Our contribution}
This paper shows that the \dsp\ for plants with six processes is undecidable.
The proof is in two steps: reduce the \PCP\ to an intermediary decision problem called the \lcp\ (Theorem~\ref{theo:reduction}) and then reduce the latter problem to the \dsp\ (Theorem~\ref{theo:lcptodsp}).
Theorem~\ref{theo:reduction} reformulates the \PCP\ as a problem of satisfaction of local constraints on a finite graph.
The proof of Theorem~\ref{theo:lcptodsp}
sheds light on the use of concurrency and causal memory to implement local constraints.

\paragraph*{Related work}
There are several classes of plants for which the \dsp\
has been shown decidable:
when the dependency graph of actions is series-parallel~\cite{gastin};
when the processes are connectedly communicating~\cite{madhu};
when the dependency graph of processes is a tree~\cite{acyclic,DBLP:conf/fsttcs/MuschollW14};
for the class of decomposable games,
which encompasses these three former decidability results~\cite{DBLP:conf/fsttcs/Gimbert17} and includes games with four processes;
and finally for plants where a single process has access to controllable actions,
while all other processes
have only access to uncontrollable actions~\cite[Corollary 7]{beutner2019translating}.

Petri games with causal memory are another, well-studied,
model of distributed asynchronous game~\cite{finkbeiner2017petri}.
In general, the number of token in a Petri game is \emph{unbounded},
and these games are undecidable, since they can encode games on vector addition systems with states~\cite[Theorem 6.9]{finkbeiner2017petri}.
When the number of tokens is \emph{bounded},
two subclasses of Petri games are known to be decidable:
Petri games with a bounded number of system players against a single environment player~\cite{finkbeiner2017petri}
as well as Petri games with a single system player against a bounded number of environment players~\cite{finkbeiner_et_al:LIPIcs:2018:8406}.
These two decidability results do not cover
the general case of Petri games with a bounded number of tokens
(in general there is strictly more
than one system player and strictly more than one environment player)
and the decidability of the general case
was until now an open question, to our knowledge.
The \dsp\ considered in the present paper
can be encoded as the existence of a winning strategy in a Petri game
with a \emph{bounded} number of tokens~\cite[Theorem 5]{beutner2019translating},
thus the undecidability result of the present paper implies as a corollary that Petri games with a bounded number of tokens are undecidable, either with the termination condition
(every token ends up in a final state) or the deadlock-freeness condition (the latter condition is discussed in the conclusion).

A recent result shows that Petri games with \emph{global} winning conditions are undecidable, even with only two system players and one environment player~\cite[Theorem 9]{finkbeiner2022global}. The latter result relies on an encoding of the \PCP\@. The global winning condition is used to enforce certain linearizations of the parallel runs and directly encode the synchronous setting of Pnueli and Rosner~\cite{pneuli1990distributed,schewe2014distributed}.
The undecidability proof of the present paper relies as well on an encoding of the \PCP,
however the winning condition is purely local, and there is no way to enforce particular linearizations of the play. Instead, we develop a specific proof technique using six processes
whose behavior is constrained by the set of possible parallel runs.
Remark that four processes would not be enough to get undecidability in our setting,
since this case is known to be decidable~\cite{DBLP:conf/fsttcs/Gimbert17}.

\paragraph*{Organisation of the paper}
Section~\ref{sec:games} defines the \dsp.
Section~\ref{sec:lcp} defines the \lcp\ and shows that the \PCP\ effectively reduces to the \lcp.
Section~\ref{sec:undec} effectively reduces the \lcp\ to the \dsp.

 \section{Distributed Asynchronous Games with Causal Memory}%
\label{sec:games}

The theory of Mazurkiewicz traces is very rich,
for a thorough presentation see~\cite{thebook}.
Here we only fix notations and recall the notions of traces, views, prime traces and parallel traces.

We fix an alphabet $A$ and a symmetric and reflexive dependency relation $D \subseteq A\times A$
and the corresponding independency relation $\ind \subseteq A\times A$ defined as $
\forall a,b\in A, (a \ind b) \iff (a,b)\not\in D$.
A  \emph{Mazurkiewicz trace}
or, more simply, a \emph{trace},
is an equivalence class
for the smallest equivalence relation $\equiv$ on $A^*$
which commutes independent letters i.e.
for all letters $a,b$
and all words $w_1,w_2$,
\[
a \ind b \implies
w_1abw_2\equiv
w_1baw_2\enspace.
\]
The words in the equivalence class are the \emph{linearizations} of the trace.
The set of all traces is denoted $A^*_\equiv$.
The trace whose only linearization is the empty word is
denoted $\epsilon$.
All linearizations of a trace $u$ have the same set of letters and length, denoted respectively $\alphabet(u)$ and $|u|$.
%Given $B\subseteq A$,
%the set of traces such that
%$\alphabet(u)\subseteq B$
%is denoted $r_\equiv^*$
%in particular

The concatenation on words naturally extends to traces.
Given two traces $u,v\in A^*_\equiv$, the trace $uv$ is the equivalence class of any word in $uv$.
The prefix relation $\pref$ is defined by
\[
(u\pref v \iff \exists w \in A^*_\equiv, uw\equiv v)\enspace,
\]
and the suffix relation is defined similarly.

\paragraph*{Maxima, prime traces and parallel traces}
A letter $a\in A$ is a \emph{maximum} of a trace $u$ if it is the last letter of one of the linearizations of $u$.
A trace $u\in A_\equiv^*$ is \emph{prime} if it has a unique maximum,
denoted $\last(u)$
and called the last letter of $u$.
Two prime traces $u$ and $v$ are said to be \emph{parallel}
if
\begin{itemize}
\item
neither
$u$ is a prefix of $v$ nor
$v$ is a prefix of $u$; and
\item
there is a trace $w$ such that both $u$ and $v$ are prefixes of $w$.
These notions are illustrated on Fig.~\ref{fig:example}.
\end{itemize}

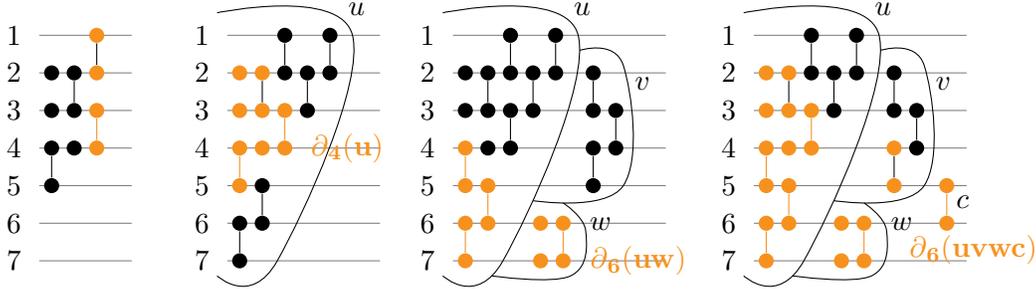
\begin{figure}

\begin{tikzpicture}
\begin{scope}[scale=1]

\newcommand{\eeee}{\mycolor}
\newcommand{\eeeeee}{\mycolor}
\newcommand{\eeeee}{1.5cm}
\newcommand{\vfour}{black}
\newcommand{\vfourr}{\mycolor}
\newcommand{\colw}{black}
\newcommand{\spec}{\vfour}
\newcommand{\ff}[1]{{}}

\newcommand{\ssss}{
\tikzstyle{every node}=[node distance=.5cm]
\node(lab1) at (1,-1) {$1$ };
\node(lab2) [below of=lab1] {$2$ };
\node(lab3) [below of=lab2] {$3$ };
\node(lab4) [below of=lab3] {$4$ };
\node(lab5) [below of=lab4] {$5$ };
\node(lab6) [below of=lab5] {$6$ };
\node(lab7) [below of=lab6] {$7$ };
\tikzstyle{every node}=[node distance=.2cm]
\node(l1) [right of=lab1] {};
\tikzstyle{every node}=[node distance=.5cm]
\node(l2) [below of=l1] {};
\node(l3) [below of=l2] {};
\node(l4) [below of=l3] {};
\node(l5) [below of=l4] {};
\node(l6) [below of=l5] {};
\node(l7) [below of=l6] {};
\tikzstyle{every node}=[node distance=\eeeee]
\node(r1) [right of=l1] {};
\draw[gray] (l1) -- (r1);
\node(r2) [right of=l2] {};
\draw[gray] (l2) -- (r2);
\node(r3) [right of=l3] {};
\draw[gray] (l3) -- (r3);
\node(r4) [right of=l4] {};
\draw[gray] (l4) -- (r4);
\node(r5) [right of=l5] {};
\draw[gray] (l5) -- (r5);
\node(r6) [right of=l6] {};
\draw[gray] (l6) -- (r6);
\node(r7) [right of=l7] {};
\draw[gray] (l7) -- (r7);

\tikzstyle{every state}=[fill=black,draw=none,inner sep=0pt,minimum size=0.2cm]
\tikzstyle{every node}=[node distance=.3cm]
\node[state](a2)[right of=l2, fill=\vfour]{};
\tikzstyle{every node}=[node distance=.5cm]
\node[state](a3)[below of=a2, fill=\vfour]{};
\node[state](a4)[below of=a3, fill=\spec]{};
\node[state](a5)[below of=a4, fill=\spec]{};
\draw[\spec] (a4) -- (a5);
\ff{\node[state](a6)[below of=a5, fill=\colw]{};
\node[state](a7)[below of=a6, fill=\colw]{};
\draw[\colw] (a6) -- (a7);}

\node[state](b2)[right of=a2, node distance=.3cm, fill=\vfour]{};
\node[state](b3)[below of=b2, fill=\vfour]{};
\draw (b2) -- (b3);
\node[state](b4)[below of=b3, fill=\vfour]{};
\ff{\node[state](b5)[below of=b4, fill=\eeeeee]{};
\node[state](b6)[below of=b5, fill=\eeeeee]{};
\draw[color=\eeeeee] (b5) -- (b6);
}

\node[state](c2)[right of=b2, fill=\eeee, node distance=.3cm]{};
\node[state](c1)[above of=c2, fill=\eeee]{};
\draw (c1) -- (c2);
\node[state](c3)[below of=c2, fill=\vfourr]{};
\node[state](c4)[below of=c3, fill=\vfourr]{};
\draw[\vfourr] (c3) -- (c4);

\ff{
\node[state](d2)[right of=c2, node distance=.3cm]{};
\node[state](d3)[below of=d2]{};
\draw (d2) -- (d3);

\node[state](e2)[right of=d2, node distance=.3cm, fill=\eeee]{};
\node[state](e1)[above of=e2, fill=\eeee]{};
\draw[color=\eeee] (e2) -- (e1);

\node(u)[above right of=e1]{$u$};

\draw [black] plot [smooth, tension=0.7] coordinates { (1.2,-0.7) (3,-1) (2,-4) (1.2,-4.2)};
}
}

\ssss

\renewcommand{\eeee}{black}
\renewcommand{\eeeee}{2.3cm}
\renewcommand{\vfour}{\mycolor}
\renewcommand{\vfourr}{\vfour}
\renewcommand{\eeeeee}{black}
\renewcommand{\ff}[1]{{#1}}

\begin{scope}[shift={(2.5,0)}]
\ssss
\node(view)[right of =c4, node distance=0.8cm]{\color{\vfour}$\bf\view_4(u)$};
\end{scope}

\renewcommand{\eeeee}{3.1cm}
\renewcommand{\vfour}{black}
\newcommand{\colv}{black}
\renewcommand{\colw}{\mycolor}
\renewcommand{\spec}{\colw}
\renewcommand{\eeeeee}{\colw}
\newcommand{\eeeeeeee}{black}

\newcommand{\sssss}{
\ssss
\node[state](f2)[right of =e2,fill=\colv]{};
\node[state](f3)[below of =f2,fill=\colv]{};
\draw[\colv] (f2)--(f3);
\node[state](f4)[below of =f3,fill=\eeeeeeee]{};
\node[state](f5)[below of =f4,fill=\eeeeeeee]{};
\draw[\colv] (f4)--(f5);
\node[state](g3)[right of =f3, node distance=.3cm]{};
\node[state](g4)[below of =g3]{};
\draw[\colv] (g3)--(g4);

\node(v)[above right of =g3]{$v$};
\draw [black] plot [smooth, tension=0.7] coordinates { (3.02,-1.2) (3.6,-1.4) (3.6,-3) (2.4,-3.2)};

\node[state](d6)
[right of =b6,fill=\colw, node distance=.7cm]{};
\node[state](d7)
[below of =d6,fill=\colw]{};
\node[state](e6)
[right of =d6,fill=\colw, node distance=.3cm]{};
\node[state](e7)
[below of =e6,fill=\colw]{};
\draw[\colw] (e6)--(e7);
\node(w)
[right of =e6, node distance=.5cm]{$w$};

\draw [black] plot [smooth, tension=0.7] coordinates { (2.8,-3.23) (3.1,-3.6) (2.9,-4.2) (1.85,-4.2)};
}

\begin{scope}[shift={(5.5,0)}]
\sssss

\node(view)[right of =e7, node distance=1cm]{\color{\colw}$\bf\view_6(uw)$};

\end{scope}

\renewcommand{\eeeee}{3.1cm}
\renewcommand{\vfour}{\mycolor}
\renewcommand{\colv}{black}
\renewcommand{\colw}{\mycolor}
\renewcommand{\spec}{\mycolor}
\renewcommand{\eeeeee}{\mycolor}
\renewcommand{\eeeeeeee}{\mycolor}

\begin{scope}[shift={(9.5,0)}]
\sssss

\node[state](oups)[right of =f5, node distance=.7cm, fill=\mycolor]{};
\node(c)[below right of =oups, node distance=.3cm]{$c$};
\node[state](oups2)[below of =oups, fill=\mycolor]{};
\draw[\mycolor] (oups)--(oups2);

\node(w)
[below right of =oups2, node distance=.5cm]{ \color{\mycolor}$\bf\view_6(uvwc)$};

\end{scope}

\end{scope}
\end{tikzpicture}

\caption{\label{fig:example}
The set of processes is $\{1\ldots 7\}$.
A letter is identified with its domain.
Here the domains are either singletons, represented by a single dot,
or pairs of contigous processes,
represented by two dots connected with a vertical segment.
The trace
$
\{2\}\{3\}\{4,5\}\{2,3\}\{4\}\{1,2\}\{3,4\}
=
\{4,5\}\{4\}\{2\}\{3\}\{2,3\}\{3,4\}\{1,2\}
$
is represented on the left-handside.
It has two maximal letters $\{1,2\}$ and $\{3,4\}$ thus is not prime.
Center left: process $4$ sees only its causal view $\view_4(u)$ (in yellow).
Center right: $uvw=uwv$ since $\dom(v)\cap \dom(w)=\emptyset$. Both $uv$ and $\view_{6}(uw)$ (in yellow) are prime prefixes of $uvw$ and they are parallel.
Right: $uv$ and $\view_{6}(uvwc)$ (in yellow) are parallel.
}
\end{figure}

%\paragraph*{Processes and games}

% !TEX root = main.tex
\subsection{Asynchronous automata}

Asynchronous automata are to traces what finite automata are to finite words, as witnessed by Zielonka's theorem~\cite{zautomata}. An asynchronous automaton is a collection of automata on finite words, whose transition tables do synchronize on certain actions.

\begin{defi}
An asynchronous automaton on alphabet $A$ with a set processes $\PP$ is a tuple
$
\mathcal{A} = ((A_p)_{p\in \PP},(Q_p)_{p\in \PP},  (i_p)_{p\in \PP},(F_p)_{p\in \PP}, \Delta)
$
where:
\begin{itemize}
\item
 every process $p\in \PP$ has a set of actions
 $A_p$, a set of states $Q_p$
 and $i_p\in Q_p$ is the initial state of $p$ and $F_p\subseteq Q_p$ its set of final states.
\item
 $A=\bigcup_{p\in \PP} A_p$.
 For every letter $a\in A$, the domain of $a$ is
 $
 \dom(a)=\{p\in \PP\mid a\in A_p\}\enspace.
$
 \item
$\Delta$ is a set of transitions of the form
$(a,(q_p,q'_p)_{p\in\dom(a)})$
where $a\in A$
and $q_p,q'_p\in Q_p$.
Transitions are \emph{deterministic}: for every $a\in A$,
if
$\delta=(a,(q_p,q'_p)_{p\in \dom(a)})\in \Delta$
and $\delta'=(a,(q_p,q''_p)_{p\in \dom(a)})\in \Delta$
then $\delta=\delta'$
(hence
 $\forall p \in \dom(a),q'_p=q''_p$).
\end{itemize}
\end{defi}

\noindent
Such an automaton works asynchronously:
each time a letter $a$ is processed, the states of the processes
in $\dom(a)$ are updated according to the corresponding
transition, while the states of other processes do not change.
This induces a natural commutation relation $\ind$ on $A$:
two letters commute iff they have no process in common i.e.
\begin{align*}
& (a \ind b) \iff (\dom(a)\cap \dom(b) = \emptyset)\enspace.
\end{align*}

The set of \emph{plays} of the automaton $\AA$
is a set of traces denoted $\plays(\AA)$ and defined inductively, along with
a mapping $\state:\plays(\AA) \to \Pi_{p\in\PP}Q_p$.
\begin{itemize}
\item $\epsilon$ is a play and $\state(\epsilon)=(i_p)_{p\in \PP}$,
\item for every play $u$ such that $(\state_p(u))_{p\in\PP}$ is defined
and $\left(a,(\state_p(u),q_p)_{p\in\dom(a)}\right)$
is a transition then
$ua$ is a play and
$
\forall p \in \PP,
\state_p(ua) =
\begin{cases}
\state_p(u) &\text{ if $p\not\in\dom(a)$,}\\
q_p &\text{ otherwise.}
\end{cases}
$
\end{itemize}

\noindent
For every play $u$, $\state(u)$ is called the \emph{global state} of $u$.
The inductive definition of $\state(u)$ is correct because it is invariant by
commutation of independent letters of $u$.

\paragraph*{The domain of a trace}
For every trace $u$ we can count how many times a process $p$
has played an action in $u$, which we denote $|u|_p$.
Formally, $|u|_p$ is first defined for words, as the length of the projection
of $u$ on $A_p$, which is invariant by commuting letters.
The domain of a trace $u$ is defined as
\[
\dom(u) = \left\{ p \in \PP \mid |u|_p \neq 0\right\}\enspace.
\]

\subsection{Processes playing against their non-deterministic environment}

Given an automaton $\AA$,
we want the processes to choose actions
in such a way that each one of them is guaranteed
to eventually reach a final state.

To take into account the fact that some actions are controllable by processes while some other actions are not, we assume that $A$ is partitioned in two disjoint sets:
\[
A=A_c \sqcup A_e
\]
where $A_c$ is the set of controllable actions and
$A_e$ the set of (uncontrollable) environment actions.
Intuitively, processes cannot prevent their environment to
play actions in $A_e$, while they can decide whether to block or allow any action in $A_c$.

We adopt a modern terminology and call the automaton $\AA$ together with the partition
$A=A_c\sqcup A_e$ a \emph{distributed game}, or even more simply a \emph{game}.
In this game, the processes form a team of players which share the same goal
but not the same information.
The plays of this game are the plays of the automaton
and the common goal of all processes is to cooperate with each other so that
the play eventually terminates, with every process in a final state.
At the beginning of the game, every process $p$
selects, among the letters $A_p$ whose domain contains $p$,
a subset of controllable actions to block, possibly all of them.
Then $p$ waits for a transition of the automaton on $A_p$.
A transition on letter $a$ is possible if either
$a$ is uncontrollable or if $a$ is controllable and is not blocked by any
of the processes in $\dom(a)$.
After the transition, each process in $\dom(a)$
updates the set of controllable actions being blocked.

%When such a transition occurs,
%$p$
%one of the  letters $a$
%.
%When , exchange information with other players taking part in the transition,
%then select a new set of controllable actions

Every process $p$ follows a \emph{distributed strategy}
which dictates how to update the set of controllable blocked actions
after every transition.
%The choice of actions by a process $p$ is dynamic: after every transition in which $p$ takes part,
%$p$ can choose a new set of controllable actions.
This choice is made by $p$ on the basis
of the information available to $p$ about the play.

The information available to $p$ is modelled by a prefix
of the global play $u$, called the \emph{causal view} of $p$ and denoted $\view_p(u)$.
In general $\view_p(u)$ is a strict prefix of $u$,
because $p$ cannot directly observe the transitions on actions
whose domain does not contain $p$. However,
$p$ can learn indirectly about such transitions:
every time $p$
performs a transition on some action $a$, he can communicate
with all other processes
participating in the transition,
i.e.\ all processes in $\dom(a)$.
%and updates its causal view.
% each transition of $p$.
%However a process can communicate with other processes
%when synchronising on a common action.
These processes update their causal view
to a common mutual value which includes
all transitions known by at least one of the processes in $\dom(a)$,
plus the current transition.
The computation of the causal view $\view_p(u)$ is
illustrated on Fig.~\ref{fig:example}
and defined formally as follows.

\begin{defi}[Causal view]\label{def:views}
For every process $p\in \PP$
and trace $u$,
the causal view of $u$ by $p$, or equivalently the \emph{$p$-view} of $u$,
denoted $\view_p(u)$,
 is the unique trace such that
 $u$ factorizes as
$
u=\view_p(u)\cdot v$ and
$v$ is the longest suffix of $u$ such that
$p \not \in \dom(v)$.
%In case $\QQ$ is a singleton $\{p\}$ the view is denoted
%$\view_p(u)$ and is either empty or prime.
\end{defi}

The strategic choices of a process are made solely on the basis of its causal view
of the global run, using a distributed strategy.

\begin{defi}[Distributed strategies, consistent and maximal plays]
Let $G=(\AA,A_c, A_e)$ be a distributed game.
A \emph{strategy for process $p$} in $G$ is a mapping
which associates with every play $u$ a set of actions $\sigma_p(u)$ such that:
\begin{itemize}
\item
all environment actions are allowed: $A_e\subseteq \sigma_p(u)$,
\item
the decision depends only on the view of the process:
$\sigma_p(u)=\sigma_p(\view_p(u))$.
\end{itemize}
A \emph{distributed strategy}
 is a tuple $\sigma=(\sigma_p)_{p\in\PP}$
where each $\sigma_p$ is a strategy of process $p$.

Let $\sigma$ be a distributed strategy.
A play $u=a_1\cdots a_{|u|}\in\plays(\AA)$ is \emph{consistent with $\sigma$},
or equivalently is a \emph{$\sigma$-play} if:
\[
\forall i \in 1\ldots |u|,
\forall p\in\dom(a_i),
a_i\in\sigma_p(a_1\cdots a_{i-1})\enspace.
\]
A $\sigma$-play is \emph{maximal} if it is not the strict prefix of another $\sigma$-play.
\end{defi}

Albeit a distributed strategy for the controllers limits which plays can occur,
by blocking some of the controllable actions,
in general there are still many different possible plays consistent with the strategy.
There are two sources of non-determinism:
first, processes cannot block
uncontrollable actions;
second the processes may strategically decide
 to allow several controllable actions,
 without knowing in advance which one may be used by the next transition.
 This is illustrated by the example on Fig.~\ref{fig:gameexample},
whose detailed analysis is provided at the end of the section.

Note that a strategy is forced to allow every environment action to be executed at every moment.
This may seem to be a huge strategic advantage for the environment.
However depending on the current state,
not every action can be effectively used in a transition
because the transition function is not assumed to be total.
So in general not every environment actions can actually occur in a play,
and from some states there may be no outgoing transition,
in which case the corresponding process terminates its computation.

 \paragraph*{Winning games}

Our goal is to synthesize  strategies
which ensure that the game terminates and
all processes are in a final state.

\begin{defi}[Winning strategy]
A strategy $\sigma$ is winning if
the set of $\sigma$-plays  is finite
and in every maximal $\sigma$-play $u$,
every process is in a final state
i.e. $
\forall p \in \PP, \state_p(u)\in  F_p\enspace.
$
\end{defi}

A winning strategy should guarantee the goal
in all maximal plays consistent with the strategy.
Thus, when one seeks to synthesise a correct controller,
the non-deterministic choice of play can be seen as another,
antagonistic player, called the environment,
trying to pick-up the worst possible transitions compatible with the strategy
and prevent the processes to achieve their goal.

In this paper however we do not equip the environment
with strategies, the only players are the processes
and we are interested in finding a strategy which guarantees
a win for the processes, whatever play is non-deterministically selected
by the environment.

\paragraph*{An example}
\begin{figure}
\includegraphics[width=10cm]{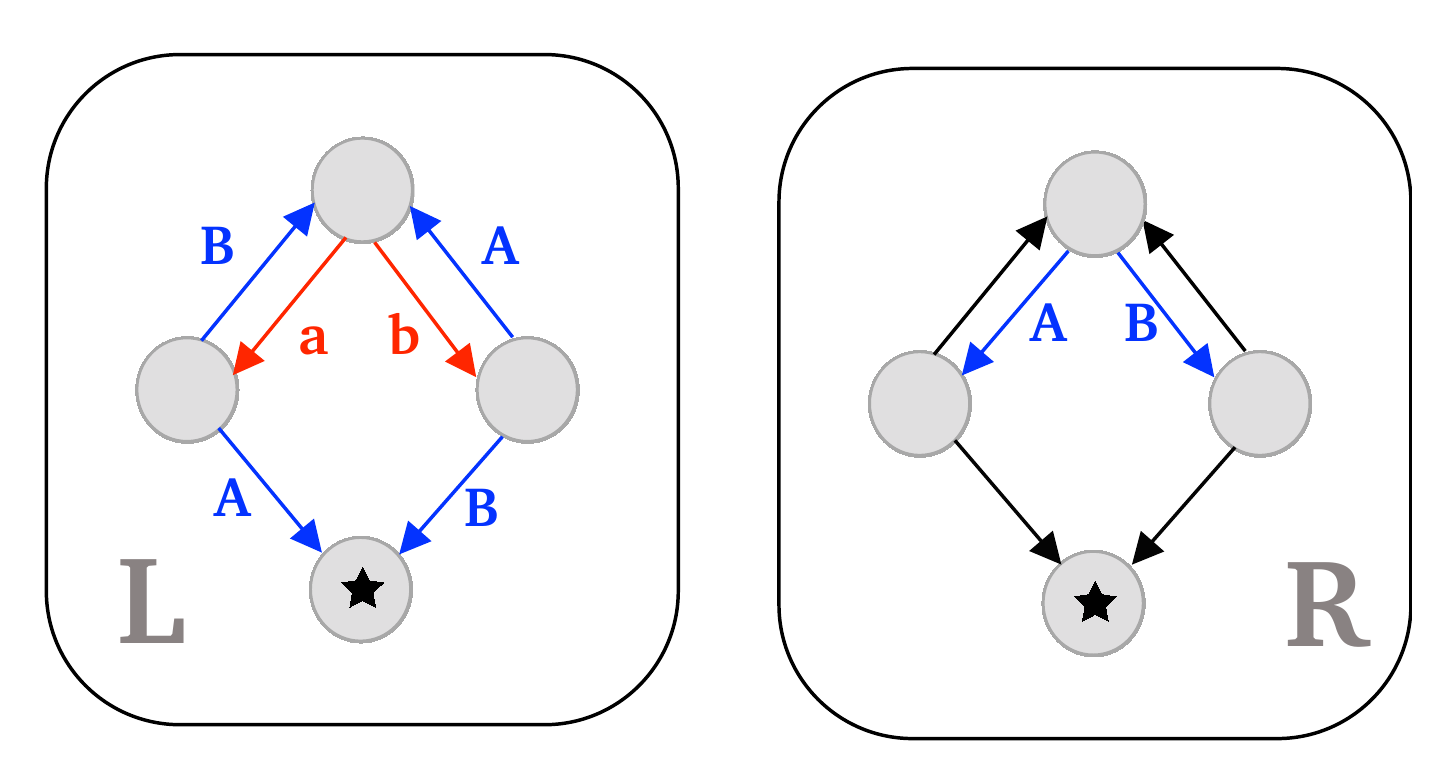}
\caption{
In this game there are two processes $L$ and $R$ (for left and right),
whose goal is to reach the starred state.
The two red actions $a$ and $b$ are uncontrollable actions, local to the $L$ process.
The two blue actions $A$ and $B$ are controllable actions, shared between processes $L$ and $R$. There are also four controllable local actions on process $R$,
depicted as black arrows.
The processes have a winning strategy.%
\label{fig:gameexample}
}
\end{figure}

We analyse the example of Fig.~\ref{fig:gameexample} in details,
in order to highlight the fact that processes may strategically need
to unlock several controllable actions at once.
Both processes start in the top state
and want to terminate by reaching the starred state on the bottom.
From the top state, the process $L$ waits for the environment
to trigger one of the uncontrollable actions $a$ or $b$, both local to $L$.
Both actions $A$ and $B$ are controllable and shared by both processes,
hence the next transition of $L$ and $R$, if any, depends on
the strategic choice of both processes.
What should process $R$ do?
He cannot observe $L$ thus his choice is independent of the transitions on $L$.

Assume first that $R$ blocks action $A$ and only allows action $B$.
This is not a very good idea.
In parallel, the environment may trigger the uncontrollable action $a$.
To avoid a deadlock in a non-final state, process $L$ should clearly not block the action $B$ and instead allow a synchronization with $R$ on action $B$.
%in which case process $L$ is back to his initial top state.
When the play $aB$ occurs,
both processes exchange their causal past,
and $R$ knows that
$L$ is back to its initial state.
Then it would a selfish losing strategy on the behalf of $R$ to
perform the controllable local transition to its final starred state,
since $R$ would terminate but $L$ would then never
be able to reach its own starred state,
and the processes would lose.
Instead, $R$ should return as well to its initial state.
As long as $R$ selects a single action,
such a loop may repeat forever,
in which case the processes (unfairly) lose the game.

What happens instead if $R$ allows both actions $A$ and $B$?
Then $L$ can wait for either of the uncontrollable actions $a$ and $b$
to occur, and allow only the action leading to the bottom starred state,
say $A$ in case $a$ has been played.
The only possible transition is then a synchronization on $A$,
in which case $R$ learns that $L$ has reached its final starred state.
Then $R$ can select the controllable local transition to its own final starred state,
and the processes win.
This is a winning strategy.

\paragraph*{The decision problem}

The following decision problem is the central motivation for
this work:

\smallskip

\medskip

{\noindent \sc \dsp:}
given a distributed game decide
whether  there exists a winning strategy.

\medskip

\smallskip

Whether this problem is decidable
has been an open problem for some time~\cite{gastin,alook,mumu}.
In this paper we show that this problem is undecidable.

\section{The \lcp}%
\label{sec:lcp}

% !TEX root = main.tex
The proof of the undecidability of the \dsp\ makes use of an intermediary decision problem:
In the sequel, we use the notation $[n]$ for the integer interval
\[
[n]=\{0,\ldots, (n-1)\}\enspace.
\]
\begin{defi}[Finite bipartite $C$-colorings]
Fix a finite set $C$ called the set of \emph{colors}.
A \emph{finite bipartite $C$-coloring}
is a function  $f : [n] \times  [m] \to C$
where $n$ and $m$ are positive integers.
The \emph{initial and final colors} of $f$ are respectively $f(0,0)$ and  $f(n-1,m-1)$.
\end{defi}

We often use the simpler term \emph{coloring}
to refer to a finite bipartite $C$-coloring.

\begin{figure}
\includegraphics[width=5cm]{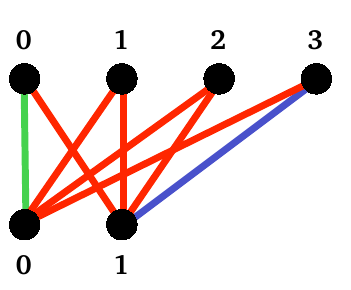}
\caption{\label{fig:triangles} A coloring $f : [n] \times  [m] \to C$ with $n=4$ and $m=2$.
The set $C$ contains three colours $R,G,B$ (red, green and blue).
The edge $(0,0)$ is green, the edge $(3,1)$ is blue and
all other edges are red.
There are three squares in $f$: the square $(G,R)=(f(0,0),f(1,1))$,
the square $(R,R)=(f(1,0),f(2,1))$ and the square $(R,B)=(f(2,0),f(3,1))$.
The pair $(G,R)$ is also both an upper-triangle $(f(0,0),f(1,0))$
and a lower-triangle  $(f(0,0),f(0,1))$.
The pair $(G,B)$ is neither a square nor a triangle,
since the corresponding edges are not adjacent.
}
\end{figure}

A coloring induces a set of patterns called
\emph{squares, upper-triangles
and lower-triangles}, this is illustrated on Fig.~\ref{fig:triangles}
and defined as follows.

\begin{defi}[Patterns induces by a coloring]
Let $f : [n] \times  [m] \to C$ be a coloring.
The \emph{patterns induced by $f$} are the following three subsets of $C^2$:
\begin{itemize}
    \item the squares of $f$ are all pairs
    $\left\{ (f(x,y),f(x+1,y+1)) \mid (x,y) \in [n-1]\times[m-1]\right\}$.
   \item the upper-triangles of $f$ are all pairs
     $\left\{ (f(x,y),f(x+1,y)) \mid (x,y) \in [n-1]\times[m]\right\}$.
  \item the lower-triangles of $f$ are all pairs
     $\left\{ (f(x,y),f(x,y+1)) \mid (x,y) \in [n]\times[m-1]\right\}$.
   \end{itemize}
\end{defi}

\noindent
The \lcp\ asks whether there exists a coloring satisfying some constraints on the initial and final colors and on the induced patterns.
A \emph{coloring constraint} is given
by three subsets $S,UT,LT$ of $C^2$, called the \emph{forbidden patterns}, and two subsets $C_i,C_f$ of $C$ called the sets of \emph{allowed initial and final} colors, respectively.
A coloring $f$ \emph{satisfies} the constraint if
its initial color
is in $C_i$, its final color is in $C_f$
and none of the patterns induced by $f$ is forbidden: no square of $f$ belongs to $S$, no upper-triangle of $f$ belongs to $UT$ and no lower-triangle of $f$ belongs to $LT$.

For example, the coloring depicted on Fig.~\ref{fig:triangles}
satisfies the constraint $C_i=\{G,R\}, C_f=\{B\}$ and $S=UT=LT=\{(B,G),(G,B)\}$.

\begin{defi}[\lcp]
Given a finite set of colors $C$ and a coloring constraint $(C_i,C_f,S,UT,LT)$,
decide whether there exists a finite bipartite $C$-coloring satisfying the constraint.
\end{defi}

\paragraph*{Two examples}
We illustrate the \lcp\  with two examples.
Set $C=\{0,+\}$ and consider the coloring constraint given
by $C_i=\{0\}$, $C_f=\{0,+\}$, $UT=\{(+,0),(0,0)\}$ $LT=\{(0,+)\}$ and $S=\emptyset$.
A coloring $f:[n]\times[m]\to C$ satisfies this constraint
iff the edges $(x,y)$ colored by $0$ are exactly those whose first coordinate is $x=0$.
The initial constraint enforces $f(0,0)=0$.
The upper-triangle constraint $UT$ prevents the symbol $0$ from appearing on any edge $(x,y)$ such that $x>0$.
By induction, the lower-triangle constraint $LT$ enforces
the symbol $0$ to appear
on any edge $(0,y)$.
One can extend this last example
on the product alphabet $\{0,+\}^2$ to enforce a coloring
to ``detect'' $0$ coordinates in both components $x$ and $y$.

Another example is $C=  \{0,-,+\}$ and
the coloring constraint given
by $C_i=\{0\}$, $C_f=\{0,-,+\}$,
$S=C^2 \setminus \{(0,0),(-,-),(+,+)\}$,
$UT= C^2 \setminus \{(0,+), (+,+), (-,0)\}$
and
$LT= C^2 \setminus \{(0,-), (-,-), (+,0)\}$.
We show that for every positive integers $n,m$,
there is a unique coloring $f:[n]\times[m]\to C$ which satisfies this constraint, and it is defined by
\[
f(x,y)=
\begin{cases}
0 & \text{ if } x = y\enspace,\\
+ & \text{ if } x > y\enspace,\\
- & \text{ if } x < y\enspace.
\end{cases}
\]
This definition of $f$ clearly satisfies the constraints,
for example the square constraint is satisfied since $x>y$ is equivalent
to $x+1 > y+1$.
To prove unicity,
remark first that $C_i$ enforces $f(0,0)=0$
and $S$ propagates this constraint to
every edge $(x,y)$ such that $x=y$.
The constraint $UT$ enforces $f(0,1)=+$ because $(0,+)$ is the only upper-triangle allowed with a $0$ on the left edge.
The $+$ is propagated by $S$ which enforces $f(x,x+1)=+$, whenever $x< \min(n,m-1)$.
And $UT$ enforces all edges $f(x,x+1),f(x,x+2),\ldots $ to be marked with $+$  because $(+,+)$ is the only upper-triangle allowed with a $+$ on the left edge.
Thus $f(x,y)=+$ whenever $x <y$. The case $x>y$ is symmetric, and the corresponding edges are marked by $-$.
If we remove the colors $-$ and $+$ from $C_f$, this creates the extra constraint $n=m$.

\medskip
The \lcp\ can encode problems far more complicated than these toy examples:
\begin{thm}\label{theo:reduction}
There is an effective reduction from the \PCP\ to the \lcp.
\end{thm}

The proof of this theorem
proceeds in two steps: first characterize a solution to \PCP\
with local rewriting rules
and second express these local rewriting rules
as coloring constraints.
Remark that the reduction generates coloring constraints which allow only solutions
$f : [n] \times  [m] \to C$ such that $n=m$ (if any).
On purpose we do not impose this constraint in the definition of the \lcp.

\begin{proof}[Proof of Theorem~\ref{theo:reduction}]
An instance of \PCP\ is a finite collection $(u_i,v_i)_{i\in I}$
of tiles on an alphabet $\Sigma$.
Each tile $(u_i,v_i)$ is a pair of non-empty words on the alphabet $\Sigma$,
$u_i$ is the top word and $v_i$ is the bottom word.
The problem is to determine whether there exists a non-empty and finite sequence of
indices in $I$ such that the concatenation of the tiles produce the same two words
on the top and the bottom.
Duplicates are allowed: the same index $i$ can appear several times in the sequence.
In case such a sequence exists, the instance
$(u_i,v_i)_{i\in I}$ is said to have a solution.
Checking whether an instance has a solution is known to be an undecidable problem~\cite{pcpundec}.

The following result characterizes the existence of a solution to an instance of \PCP\@.

\begin{lem}[Local characterization of a PCP solution]\label{lem:caracpcp}
Let $(u_i,v_i)_{i\in I}$ be an instance of \PCP\ on the alphabet $\Sigma$. Let  $i_0,i_1,\ldots,i_k$ and
$j_0,j_1,\ldots,j_\ell$ be non-empty and finite sequences of indices
and $u= u_{i_0}\ldots u_{i_k}$
and $v= v_{j_0}\ldots v_{j_\ell}$
the corresponding tiles concatenations.
Denote $X=[|u|]\times[|v|]$.
    %Denote $\phi_u$ the mapping from ${0,\ldots,|u|-1}$ to $I$ such that which associates with $\phi_u(|u|)=\sharp$ and for every other index $0<= x <|u|$, $\phi_u(x)$ is the index $i_x\in I$ of the word $u_i$ appearing at index $x$ in the factorisation $u= u_{i_0}\ldots u_{i_k}$. Define $\phi_v$ symetrically, with respect to $v$.

The equality $u=v$ holds if and only if there exists a subset $\yes$ of $X$
    with the following properties.
    \begin{enumerate}
    \item
    $\yes$ contains $(|u|-1,|v|-1)$.
    \item
    Let $x=(x_u,x_v)\in \yes $
    such that $x_u>0$ and $x_v > 0$.
    Then
$(x_u -1,x_v -1) \in \yes$.
    \item
    $(0,0)$ belongs to $\yes$
    and this is the only element
    $x=(x_u,x_v)\in \yes$
    with either coordinate equal to $0$.
    \item
    For every $x=(x_u,x_v)\in \yes$ the letter of index $x_u$ in $u$ is the same as the letter of index $x_v$ in $v$.
    \end{enumerate}
The equality $i_0,i_1,\ldots,i_k = j_0,j_1,\ldots,j_\ell$
holds
if and only if there exists a subset $\yess$ of $X$ with the following properties:
\begin{enumerate}
       \item[(5)]
    $\yess$ contains $(|u|-1,|v|-1)$.
    \item[(6)]
     Let $x=(x_u,x_v)\in \yess$
    such that $x_u >0$ and
$x_v >0$.
Say that $x$ \emph{starts a new tile in $u$} if $x_u$ is the first index of a tile in the factorisation
$u= u_{i_0}\ldots u_{i_k}$
i.e.\ if there exists $0\leq m \leq k$
such that $x_u=|u_{i_0} \ldots u_{i_m}|$.
Say that $x$ \emph{starts a new tile in $v$} if the symmetrical condition holds for $x_v$ and
$v= v_{i_0}\ldots v_{j_\ell}$.
If $x$ starts a new tile in both $u$ and $v$ then $(x_u-1,x_v-1)\in \yess$.
If $x$ does \emph{not} start a new tile in $u$ then $(x_u-1,x_v)\in \yess$. If $x$ does \emph{not} start a new tile in $v$ then $(x_u,x_v-1)\in \yess$.
\item[(7)]
$(0,0)$ belongs to $\yess$.
For every $(x_u,x_v)\in \yess$,
($x_u < |u_{i_0}| \iff x_v < |v_{j_0}| $).
\item[(8)]
     Let $x=(x_u,x_v)\in \yess$ and denote $i_u(x)\in I$ the index of the tile
    appearing at index $x_u$ in the factorisation $u= u_{i_0}\ldots u_{i_k}$.
    Define $i_v(x)\in I$ symmetrically with respect to
    the factorisation $v= v_{i_0}\ldots v_{j_\ell}$.
    Then $i_u(x)$ and $i_v(x)$ are equal.

\end{enumerate}
\end{lem}
\begin{proof}
The proof is elementary.

We prove the first part of the lemma: the characterisation of $u=v$.
For the direct implication, assume $u=v$.
Then the set  $\yes = \{ (x,x)  \mid 0\leq x < |u|=|v|\}$ has properties 1.\ to 4.
In the opposite direction, assume
there exists $\yes \subseteq X$ with properties 1.\ to 4.
We show that:
\begin{align}%
\label{eq:ij}
\text{
for every $0< i\leq \min(|u|,|v|)$,
$\yes$ contains $(i,i)$
}\enspace.
\end{align}
 %$\yes$ contains all pairs $(i,j) \in X$ such that $i=j$.
Assume w.l.o.g.\ that $|u|=\min(|u|,|v|)$.
A simple induction shows that for every $0< i\leq \min(|u|,|v|)$,
$\yes$ contains $(|u|-i,|v|-i)$:
the case $i=1$ follows from property 1.\ and the induction step from property 2. Then $(0,|v|-|u|)\in \yes$
thus, by property 3, $|u|=|v|$.
Hence property~\eqref{eq:ij}.
Finally $u=v$ follows by property 4.

We now prove the second part of the lemma: the characterisation of $i_0,i_1,\ldots,i_k = j_0,j_1,\ldots,j_\ell$.
The proof relies on the subset $A \subseteq X$
containing all pairs $(x_u,x_v)\in X$ such that the number of tiles appearing \emph{after} position $x_u$ in $u$
is the same that the number of tiles appearing \emph{after} position $x_v$ in $v$. Formally,
\begin{multline*}
A = \{ (x_u,x_v)\in X\mid
\exists m \in 0,\ldots, \min(k,\ell)\\
|u_{i_0}\cdots u_{i_{k - m-1}}|
\leq x_u
< |u_{i_0}\cdots u_{i_{k-m}}|
\text{ and }
|v_{j_0}\cdots v_{j_{\ell - m-1}}|
\leq  x_v
< |v_{j_0}\cdots v_{j_{\ell  - m}}|
\}
\enspace.
\end{multline*}
Remember that all the words $u_i,i\in I$ and $v_i,i\in I$ appearing in the tiles are non-empty.
Thus, $A$ contains $(|u|-1,|v|-1)$. And $A$ contains $(0,0)$
iff $k=\ell$.
Following its definition parametrized by $m\in 0\ldots \min(k,\ell)$, $A$ is partitioned in $A_0, \ldots, A_{\min(k,\ell)}$.

For the direct implication,
Assume $i_0,i_1,\ldots,i_k = j_0,j_1,\ldots,j_\ell$ (hence in particular $k=\ell$).
Then the set $\yess=A$
clearly satisfies all properties (5) to (8).

\medskip

Conversely, assume $\yess\subseteq X$ satisfies all properties (5) to (8), and show that $i_0,i_1,\ldots,i_k = j_0,j_1,\ldots,j_\ell$.

We show first that $\yess$ contains all elements in $A$.
Assume by contradiction that some element $x$ in $A$ is missing from $\yess$, and assume $x$ is maximal for the lexicographic ordering in $A\setminus \yess$.
Let $m$ such that $x \in A_m$.
The set $A_m$ is the cartesian product of two intervals,
hence it has a maximum, namely
$x_m=(|u_{i_0} \ldots u_{i_{k - m}}|-1, |u_{i_0} \ldots u_{i_{\ell - m}}|-1)$.
There are two cases, depending whether $x=x_m$ or not.
Assume first $x=x_m$.
According to property (5),
the maximal element $(|u|-1,|v|-1)$ in $A$
belongs to $\yess$  thus $x=x_m \neq (|u|-1,|v|-1)$.
Since $(|u|-1,|v|-1) \in A_0$ then  $m > 0$. As a consequence, $(x_u+1,x_v+1)$ is the minimal
element in $A_{m-1}$ and by induction hypothesis,
$(x_u+1,x_v+1)\in \yess$.
Moreover,  $(x_u+1,x_v+1)$ starts a new tile in both $u$ and $v$
and by property (6), we get $x\in \yess$ as well.
The second case is when $x$ is not maximum in $A_m$.
Since $A_m$ is the cartesian product of two intervals, either
$(x_u+1,x_v)$ or $(x_u,x_v+1)$  is in $A_m$ as well, and $x$ does not start a new tile on the corresponding coordinate. We conclude that $x\in\yess$ using the induction hypothesis and property (6).

Now we show $k=\ell$.
Assume w.l.o.g. $k=\min(k,\ell)$.
Since $\yess$ contains all $A$ then in particular $\yess$ contains
$(0,|v_{j_0}\cdots v _{j_{\ell - k}}| -1)$ (take $m=k$ in the definition of $A$).
According to property (7),
$|v_{j_0}\cdots v_{j_{\ell - k}}| \leq |v_{j_0}| $ thus $k=\ell$.

Using property (8), we conclude  that $i_0,i_1,\ldots,i_k = j_0,j_1,\ldots,j_\ell$.
\end{proof}

We now define an instance of the \lcp\ encoding the constraints in Lemma~\ref{lem:caracpcp}. We start with defining the set of colors.
Let $Q_u$ be the finite subset of $\Sigma\times I \times \NN\times\{0,1\}^2$
containing all tuples $(a,i,m,b_0,b_1)$
where $m < |u_i|$ and $a$ is the $m$-th letter of $u_i$.
Given $q=(a,i,m,b_0,b_1)\in Q_u$,
$a$ is called the letter of $q$, $i$ its tile index, $m$ its tile position, $b_0$ its initial letter flag and $b_1$ its initial tile flag.
The tile position is maximal in $q$ if it is equal to $|u_i|-1$.
Define $Q_v$ similarly with respect to the bottom parts of the tiles $(v_i)_{i\in I}$.
Let $C$ be the set of colors
\[
Q_u\times Q_v \times \{\yes,\no\}\times\{\yess,\noo\}
\enspace.\]

We define a constraint coloring problem reflecting conditions (1) to (8) of the previous lemma. We use $*$ as a wildcard character which can be replaced by any symbol.

First, we set constraints
which enforce the
$Q_u$-component of $f(x_u,x_v)$
to depend only on $x_u$
and
the $Q_v$-component of $f(x_u,x_v)$
to depend only on $x_v$.
For that we include in the set of forbidden
lower-triangle $LT$ all pairs
$(q_u,*,*,*)(q_u',*,*,*)$
with $q_u\neq q_u'$.
A simple induction shows that
if a coloring $f:[n]\times[m]\to C$ does not induce such lower-triangles pattern, then  there exists $f_u : [n] \to Q_u$
such that
 for every $x=(x_u,x_v)\in[n]\times[m] $,
 the first component of $f(x)$
 is $f_u(x_u)$.
 Symmetrically,
 we add to the set of forbidden upper-triangles $UT$ all pairs
 $(*,q_v,*,*)(*,q_v',*,*)$
with $q_v\neq q_v'$
so that the $Q_v$-component of $f$
is induced by  $f_v : [m] \to Q_v$.

Denote $u\in \Sigma^n$ to be the word
 of length $n$ whose $x_u$-th letter
 is the first component of $f_u(x_u)\in Q_u$.
 Define $v\in \Sigma^m$ symmetrically with respect to $f_v$.

 Second, we want to ensure that $u$ is a concatenation
 of upper tiles of the PCP instance,
 i.e.
 there exists $i_0,i_1,\ldots , i_k$ such that
$u = u_{i_0} \cdots u_{i_k}$.
This regular constraint can be checked by
a deterministic automaton on $Q_u$,
which we implement using the coloring constraints.
First, we require that in every initial color $(q,*,*,*)$, the tile position in $q$ is $0$.
Moreover, we put constraints on all upper-triangles
$(q,*,*,*),(q',*,*,*)$:
either the tile position is maximal in $q$ and has value $0$ in $q'$, or the tile index does not change between $q$ and $q'$ while the tile position is incremented of exactly one unit.
Finally, we require that in every final color $(q,*,*,*)$, the tile position in $q$ is maximal.
With a symmetric construction,
we ensure  the symmetric constraint on $v$:
there exists $j_0,j_1,\ldots , j_\ell$ such that
$v = v_{j_0} \cdots v_{j_\ell}$.

 Third, we want that the initial letter flag
 and the initial tile flag of
$f_u(x_u)$ indicate respectively whether or not $x_u=0$
and
 whether or not $x_u < |u_{i_0}|$.
We require that all initial colors have both flags set to $1$.
We forbid any upper triangle $(q_u,*,*,*)(q_u',*,*,*)$ where $q_u'$ has the initial letter flag set to $1$.
We forbid any upper triangle $(q_u,*,*,*)(q_u',*,*,*)$
where $q_u'$ has the initial tile flag set to $1$
unless $q_u$ has also the initial tile flag set to $1$
and the tile position in $q_u$ is not maximal.
A symmetric construction ensures that
the initial letter and tile flags of
$f_v(u_v)$ indicate respectively whether or not $u_v=0$ and $u_v< |v_{j_0}|$.

This way we have a natural correspondence between finite bipartite coloring satisfying the above constraints
and pairs of factorisations
$u = u_{i_0} \cdots u_{i_k}$ and
$v = v_{j_0} \cdots v_{j_\ell}$
like in Lemma~\ref{lem:caracpcp}.
Using this correspondence,
we can rephrase the eight conditions
in Lemma~\ref{lem:caracpcp} as eight coloring constraints on
a set of colours $C'\subseteq C$.
\begin{enumerate}
\item
Every final colour has type $(*,*,\yes,*)$.
\item
Forbid all squares $(*,*,\no,*),(*,*,\yes,*)$.
\item
Every initial colour has type $(q_u,q_v,\yes,*)$, with both initial letter flags in $q_u$ and $q_v$ set to $1$.
Remove from the alphabet $C$ any color
$(q_u,q_v,\yes,*)$
such that the initial letter flags of $q_u$ and $q_v$ are different.
\item
Remove from the alphabet $C$ any color
$(q_u,q_v,\yes,*)$
such that the letters in $q_u$
and $q_v$ (i.e.\ their projections on $\Sigma$) are different.
\item
Every final colour has type $(*,*,*,\yess)$.
\item
Forbid all squares $(*,*,*,\noo)$,$(q_u',q_v',*,\yess)$
such that both tile positions in $q_u'$ and $q_v'$ are $0$.
Forbid all upper-triangles $(*,*,*,\noo)$,$(q'_u,*,*,\yess)$
unless the tile position in $q'_u$ is $0$.
Forbid all lower-triangles $(*,*,*,\noo)$,$(*,q_v',*,\yess)$
unless the tile position in $q_v'$ is $0$.
\item
Every initial colour has type $(q_u,q_v,*,\yess)$, with both initial tile flags in $q_u$ and $q_v$ set to $1$.
Remove from the alphabet $C$ any color $(q_u,q_v,*,\yess)$
such that the initial tile flags in $q_u$ and $q_v$ are different.
\item
Remove from the alphabet $C$ any color
$(q_u,q_v,*,\yess)$
such that the tile indices in $q_u$
and $q_v$ are different.
\end{enumerate}

\noindent
Applying conditions,
3, 4, 7 and 8, we get a set of colours $C'\subset C$
on which the constraint coloring problem is defined.
The initial condition $C_i$ is the set of colours in $C'$ of type
$(q_u,q_v,\yes,\yess)$
such that in both
$q_u$ and $q_v$,
the tile position is $0$
and both initial letter flags and initial tile flags are set to $1$.
The final condition $C_f$ is the set of colours
of type $(q_u,q_v,\yes,\yess)$ where the tile positions in both $q_u$
and $q_v$ are maximal.
The set of forbidden patterns is obtained by taking the union of all forbidden squares, upper-triangles and lower-triangles mentioned
above.

According to Lemma~\ref{lem:caracpcp},
there exists a coloring satisfies the constraints
if and only if the PCP instance has a solution.
This terminates the proof of Theorem~\ref{theo:reduction}.
\end{proof}

A direct consequence of Theorem~\ref{theo:reduction}
is that
the \lcp\ is undecidable.

\section{Undecidability of the Distributed Synthesis Problem with \texorpdfstring{$6$}{6} processes}%
\label{sec:undec}

The main result of this section is:
\begin{thm}\label{theo:lcptodsp}
There is an effective reduction from the \lcp\
to the \dsp\ with six processes.
\end{thm}

We show how to effectively transform
an instance $(C_i,C_f,S,UT,LT)$ of the \lcp\ on a set of colors $C$
into a distributed game $G$
such that:
\begin{lem}\label{lem:equiv}
There is a winning strategy in $G$ if and only if
there is a finite bipartite coloring satisfying the constraints $(C_i,C_f,S,UT,LT)$.
\end{lem}
In the rest of the section we describe the construction of the game $G$ and then prove
Lemma~\ref{lem:equiv}.

\subsection{Turning a coloring problem into a game}

The game $G$ has six processes divided into two pools: the top pool $T_0,T_1,T_2$ and the bottom pool $B_0,B_1, B_2$.
The transitions of these processes are quite similar,
thus we often use the notation $X_\ell$ to designate any of the six processes,
a generic notation where $X\in\{T,B\}$ and $\ell\in\{0,1,2\}$.
For every process $X_\ell$ we denote $\nex(X_\ell)$
the process $X_{(\ell +1) \mod 3}$
and we denote
$\prev(X_\ell)$
the process $X_{(\ell -1) \mod 3}$.
For example, $\nex(T_0)=T_1$ and $\prev(B_0)=B_2$.

\paragraph*{Uninterrupted plays.}
These rare plays where only the following actions are used:

\begin{itemize}
\item
There are controllable actions called \emph{increments},
which synchronize two processes of the same pool.
In each pool $X\in\{T,B\}$
there are three increment actions
$\increment_{X,0},\increment_{X,1}$ and $\increment_{X,2}$.
Every increment $\increment_{X,\ell}$
is controllable and
is shared between processes $X_\ell$ and $\nex(X_{\ell})$.
\item
Every process $X_\ell$ has also a controllable local  action $\terminate_{X,\ell}$.
\item
There is a controllable action \win\ synchronizing all six processes.
\end{itemize}

\noindent
The transitions are defined so that every maximal uninterrupted play
looks like the one on Figure~\ref{fig:uninterrupted}:
it
 is the parallel product of two plays
$(\increment_{T,0}\increment_{T,1}\increment_{T,2})^n$
and
$(\increment_{B,0}\increment_{B,1}\increment_{B,2})^m$
for some positive integers $n,m$,
followed by the six local $\terminate$ actions and finally the global action $\win$.

Since all the actions in an uninterrupted play are controllable,
the number of rounds performed before playing actions \terminate\ is
determined by the players
and is not influenced by
the environment.
Intuitively, when the number of rounds of the top and bottom pools are  $n$ and $m$, respectively,
the players claim to have a solution $f:[n]\times [m] \to C$ to the constrained coloring problem.

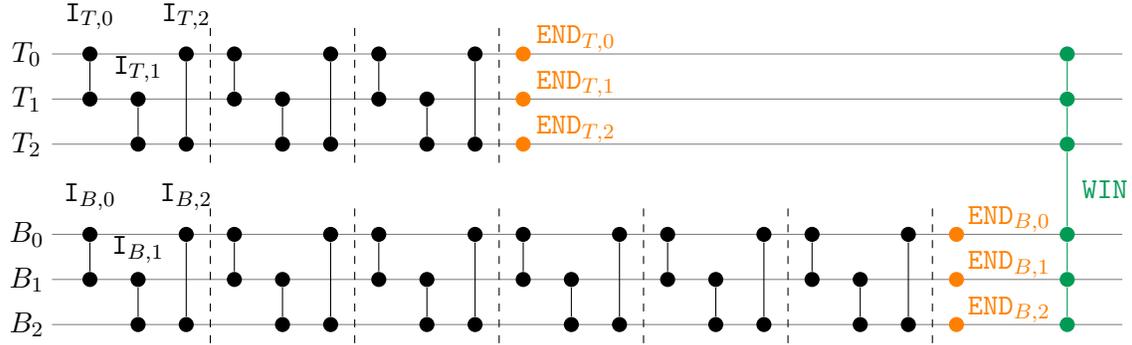
\begin{figure}[ht]
    \begin{tikzpicture}
\begin{scope}[scale=1]

\newcommand{\eeee}{mycolor}
\newcommand{\eeeeee}{mycolor}
\newcommand{\eeeee}{1.5cm}
\newcommand{\vfour}{black}
\newcommand{\vfourr}{mycolor}
\newcommand{\colw}{black}
\newcommand{\colwin}{ForestGreen}
\newcommand{\colend}{orange}
\newcommand{\spec}{\vfour}
\newcommand{\ff}[1]{{}}

\newcommand{\vdist}{.6cm}
\newcommand{\hdist}{.64cm}

\tikzstyle{every node}=[node distance=\vdist]
\node(lab1) at (1,-1) {$T_0$ };
\node(lab2) [below of=lab1] {$T_1$ };
\node(lab3) [below of=lab2] {$T_2$ };
\node(lab0) [below of=lab3] { };
\node(lab4) [below of=lab0] {$B_0$ };
\node(lab5) [below of=lab4] {$B_1$ };
\node(lab6) [below of=lab5] {$B_2$ };

\tikzstyle{every node}=[node distance=.2cm]
\node(l1) [right of=lab1] {};
\node(l2) [right of=lab2] {};
\node(l3) [right of=lab3] {};
\node(l4) [right of=lab4] {};
\node(l5) [right of=lab5] {};
\node(l6) [right of=lab6] {};

%%%%% the horizontal lines
\tikzstyle{every node}=[node distance=14.5cm]
\node(r1) [right of=l1] {};
\draw[gray] (l1) -- (r1);
\node(r2) [right of=l2] {};
\draw[gray] (l2) -- (r2);
\node(r3) [right of=l3] {};
\draw[gray] (l3) -- (r3);
\node(r4) [right of=l4] {};
\draw[gray] (l4) -- (r4);
\node(r5) [right of=l5] {};
\draw[gray] (l5) -- (r5);
\node(r6) [right of=l6] {};
\draw[gray] (l6) -- (r6);

\newcommand{\oneround}{
\tikzstyle{every state}=[fill=black,draw=none,inner sep=0pt,minimum size=0.2cm]
\tikzstyle{every node}=[node distance=\hdist]
\node[state](a0)[fill=\vfour] at (0,0) {};
\node(dummy)[right of=a0] {};
\node[state](c0)[right of=dummy] {};
\tikzstyle{every node}=[node distance=\vdist]
\node[state](a1)[below of=a0]{};
\draw[\colw] (a0) -- (a1);
\tikzstyle{every node}=[node distance=\hdist]
\node[state](b1)[right of=a1]{};
\tikzstyle{every node}=[node distance=\vdist]
\node[state](b2)[below of=b1]{};
\draw[\colw] (b1) -- (b2);
\tikzstyle{every node}=[node distance=\hdist]
\node[state](c2)[right of=b2]{};
\draw[\colw] (c0) -- (c2);
%%%%the separation ation
\tikzstyle{every node}=[node distance=0.5*\hdist]
\node(sep0)[right of=c0]{};
\node(sep1)[right of=c2]{};
\tikzstyle{every node}=[node distance=0.8*\vdist]
\node(sepa0)[above of=sep0]{};
\node(sepa1)[below of=sep1]{};
\draw[dashed, \colw] (sepa0) -- (sepa1);
}

\newcommand{\totoend}{
\tikzstyle{every state}=[fill=\colend,draw=none,inner sep=0pt,minimum size=0.2cm]
\tikzstyle{every node}=[node distance=\vdist]
\node[state](e0)at (0,0) {};
\node[state](e1)[below of=e0] {};
\node[state](e2)[below of=e1]{};
}

\newcommand{\totowin}{
\tikzstyle{every state}=[fill=\colwin,draw=none,inner sep=0pt,minimum size=0.2cm]
\tikzstyle{every node}=[node distance=\vdist]
\node[state](e0)at (0,0) {};
\node[state](e1)[below of=e0] {};
\node[state](e2)[below of=e1]{};
\node(dummy)[below of=e2]{};
\node[state](eb0)[below of=dummy] {};
\node[state](eb1)[below of=eb0] {};
\node[state](eb2)[below of=eb1]{};
\draw[\colwin] (e0) -- (e1) -- (e2) -- (eb0) -- (eb1) -- (eb2) ;
}

\begin{scope}[shift={($(l4)+(\hdist,0)$)}]
\oneround
\node(lab1) at (0,0.5) {$\increment_{B,0}$ };
\node(lab2) at (\hdist,-0.2) {$\increment_{B,1}$ };
\node(lab3) at (2 * \hdist,0.5 ) {$\increment_{B,2}$ };
\begin{scope}[shift={($(3*\hdist,0)$)}]
\oneround
\begin{scope}[shift={($(3*\hdist,0)$)}]
\oneround
\begin{scope}[shift={($(3*\hdist,0)$)}]
\oneround
\begin{scope}[shift={($(3*\hdist,0)$)}]
\oneround
\begin{scope}[shift={($(3*\hdist,0)$)}]
\oneround
\begin{scope}[shift={($(3*\hdist,0)$)}]
\totoend
\node(lab1) at (0.7cm,0.2cm) {$\textcolor{\colend}{\terminate_{B,0}}$ };
\node(lab2) at (0.7cm,0.2cm - \vdist) {$\textcolor{\colend}{\terminate_{B,1}}$ };
\node(lab3) at (0.7cm,0.2cm - \vdist- \vdist) {$\textcolor{\colend}{\terminate_{B,2}}$ };
%\node(lab1) at (0.5,0.4) {$\textcolor{\colend}{\terminate_{T,0}}$ };
\begin{scope}[shift={($(2.3*\hdist,4*\vdist)$)}]
\totowin
\node(lab1) at (0.5,-3*\vdist) {$\textcolor{\colwin}{\win}$ };
\end{scope}
\end{scope}
\end{scope}
\end{scope}
\end{scope}
\end{scope}
\end{scope}
\end{scope}

\begin{scope}[shift={($(l1)+(\hdist,0)$)}]
\oneround
\node(lab1) at (0,0.5) {$\increment_{T,0}$ };
\node(lab2) at (\hdist,-0.2) {$\increment_{T,1}$ };
\node(lab3) at (2 * \hdist,0.5 ) {$\increment_{T,2}$ };

\begin{scope}[shift={($(3*\hdist,0)$)}]
\oneround
\begin{scope}[shift={($(3*\hdist,0)$)}]
\oneround
\begin{scope}[shift={($(3*\hdist,0)$)}]
\totoend
\node(lab1) at (0.7cm,0.2cm) {$\textcolor{\colend}{\terminate_{T,0}}$ };
\node(lab2) at (0.7cm,0.2cm - \vdist) {$\textcolor{\colend}{\terminate_{T,1}}$ };
\node(lab3) at (0.7cm,0.2cm - \vdist- \vdist) {$\textcolor{\colend}{\terminate_{T,2}}$ };
\end{scope}
\end{scope}
\end{scope}
\end{scope}

\end{scope}
\end{tikzpicture}

    \caption{A maximal uninterrupted play where the top pool plays $3$ rounds and the bottom pool plays $6$ rounds. Rounds are separated by dashes. All actions are controllable. }%
    \label{fig:uninterrupted}
\end{figure}

In each pool $X\in\{T,B\}$, increments have to be played in a fixed order:  \[\increment_{X,0},\increment_{X,1},\increment_{X,2},\increment_{X,0},,\increment_{X,1},\ldots\enspace.\]
Note that
no two increment actions of the same pool can commute,
thus an interrupted play $(\increment_{X,0}\increment_{X,1}\increment_{X,2})^*$
in a pool $X$
has a single linearization.
Each subword $\increment_{X,0},\increment_{X,1},\increment_{X,2}$ is called a \emph{round}.
This pattern is enforced by maintaining in the state space of every process $X_\ell$ a flag indicating whether or not the process is allowed to synchronize with
$\nex(X_\ell)$. The flag is initially set to $1$ for $X_0$ and to $0$ for $X_1$ and $X_2$, and is toggled every time the process performs an increment.

The only way for processes to win
is a transition on the global action \win, in which case they immediately and definitively enter
a final state, and no further transition can occur after that.
A transition on \win\ can be triggered once all six processes have performed a local transition on their \terminate\ action.
%Once all six processes are in the state \sss,
%the transition \win\ can be triggered, with no further conditions.

The transition on the local action $\terminate_{X,\ell}$ for process $X_\ell$ is available if and only if
the process has played at least one increment action and has finished the current round. For $X_0$ and $X_2$ (resp.\ for  $X_1$) it means that the last increment was $\increment_{X_2}$ (resp.\ was $\increment_{X_1}$). This constraint can easily be encoded with a flag in the state space of the process.

A play stays uninterrupted,
 like the one on Figure~\ref{fig:uninterrupted},
 when all actions are controllable and
 the  environment does not play at all,
  in which case the game is easy to win.
 But in general,
 a winning strategy
  should also react correctly
  to uncontrollable actions
 of the environment
called
 \emph{checks}.

\paragraph*{Interrupting plays by color checks.}

The environment has the ability to interrupt a play
by triggering an uncontrollable \emph{check action}.
This is represented on Figure~\ref{fig:check}.
Intuitively, when the environment interrupts the
play after $x$ rounds of the top pool and $y$ rounds of the bottom pool,
the players are asked to pick the color $f(x,y)\in C$ of the edge $(x,y)$ of the
solution to the coloring problem.
There are rules which enforce the players to
define that way a coloring $f: [n]\times [m] \to C$ which satisfies the constraints,
where $n$ and $m$ are the number of rounds of the top and bottom
pools, respectively, before playing $\terminate$. If the strategy of the players makes them cheat or describe a coloring $f$ which does not satisfy the constraints, the environment can trigger one or two checks which make the players lose.

\begin{center}
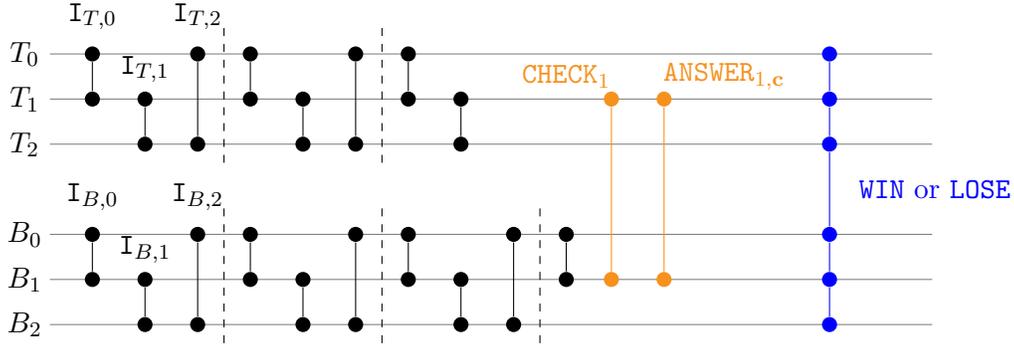
\begin{figure}[ht]
    \begin{tikzpicture}
\begin{scope}[scale=1]

\newcommand{\eeee}{mycolor}
\newcommand{\eeeeee}{mycolor}
\newcommand{\eeeee}{1.5cm}
\newcommand{\vfour}{black}
\newcommand{\vfourr}{mycolor}
\newcommand{\colw}{black}
\newcommand{\colwin}{ForestGreen}
\newcommand{\colend}{orange}
\newcommand{\colcheck}{BurntOrange}
\newcommand{\spec}{\vfour}
\newcommand{\ff}[1]{{}}

\newcommand{\vdist}{.6cm}
\newcommand{\hdist}{.7cm}

\tikzstyle{every node}=[node distance=\vdist]
\node(lab1) at (1,-1) {$T_0$ };
\node(lab2) [below of=lab1] {$T_1$ };
\node(lab3) [below of=lab2] {$T_2$ };
\node(lab0) [below of=lab3] { };
\node(lab4) [below of=lab0] {$B_0$ };
\node(lab5) [below of=lab4] {$B_1$ };
\node(lab6) [below of=lab5] {$B_2$ };

\tikzstyle{every node}=[node distance=.2cm]
\node(l1) [right of=lab1] {};
\node(l2) [right of=lab2] {};
\node(l3) [right of=lab3] {};
\node(l4) [right of=lab4] {};
\node(l5) [right of=lab5] {};
\node(l6) [right of=lab6] {};

%%%%% the horizontal lines
\tikzstyle{every node}=[node distance=12cm]
\node(r1) [right of=l1] {};
\draw[gray] (l1) -- (r1);
\node(r2) [right of=l2] {};
\draw[gray] (l2) -- (r2);
\node(r3) [right of=l3] {};
\draw[gray] (l3) -- (r3);
\node(r4) [right of=l4] {};
\draw[gray] (l4) -- (r4);
\node(r5) [right of=l5] {};
\draw[gray] (l5) -- (r5);
\node(r6) [right of=l6] {};
\draw[gray] (l6) -- (r6);

\newcommand{\oneround}{
\tikzstyle{every state}=[fill=black,draw=none,inner sep=0pt,minimum size=0.2cm]
\tikzstyle{every node}=[node distance=\hdist]
\node[state](a0)[fill=\vfour] at (0,0) {};
\node(dummy)[right of=a0] {};
\node[state](c0)[right of=dummy] {};
\tikzstyle{every node}=[node distance=\vdist]
\node[state](a1)[below of=a0]{};
\draw[\colw] (a0) -- (a1);
\tikzstyle{every node}=[node distance=\hdist]
\node[state](b1)[right of=a1]{};
\tikzstyle{every node}=[node distance=\vdist]
\node[state](b2)[below of=b1]{};
\draw[\colw] (b1) -- (b2);
\tikzstyle{every node}=[node distance=\hdist]
\node[state](c2)[right of=b2]{};
\draw[\colw] (c0) -- (c2);
%%%%the separation ation
\tikzstyle{every node}=[node distance=0.5*\hdist]
\node(sep0)[right of=c0]{};
\node(sep1)[right of=c2]{};
\tikzstyle{every node}=[node distance=0.8*\vdist]
\node(sepa0)[above of=sep0]{};
\node(sepa1)[below of=sep1]{};
\draw[dashed, \colw] (sepa0) -- (sepa1);
}

\newcommand{\totoend}{
\tikzstyle{every state}=[fill=\colend,draw=none,inner sep=0pt,minimum size=0.2cm]
\tikzstyle{every node}=[node distance=\vdist]
\node[state](e0)at (0,0) {};
\node[state](e1)[below of=e0] {};
\node[state](e2)[below of=e1]{};
}

\newcommand{\totowin}{
\tikzstyle{every state}=[fill=blue,draw=none,inner sep=0pt,minimum size=0.2cm]
\tikzstyle{every node}=[node distance=\vdist]
\node[state](e0)at (0,0) {};
\node[state](e1)[below of=e0] {};
\node[state](e2)[below of=e1]{};
\node(dummy)[below of=e2]{};
\node[state](eb0)[below of=dummy] {};
\node[state](eb1)[below of=eb0] {};
\node[state](eb2)[below of=eb1]{};
\draw[blue] (e0) -- (e1) -- (e2) -- (eb0) -- (eb1) -- (eb2) ;
}

\begin{scope}[shift={($(l4)+(\hdist,0)$)}]
\oneround
\node(lab1) at (0,0.5) {$\increment_{B,0}$ };
\node(lab2) at (\hdist,-0.2) {$\increment_{B,1}$ };
\node(lab3) at (2 * \hdist,0.5 ) {$\increment_{B,2}$ };
\begin{scope}[shift={($(3*\hdist,0)$)}]
\oneround
\begin{scope}[shift={($(3*\hdist,0)$)}]
\oneround
\begin{scope}[shift={($(3*\hdist,0)$)}]
\tikzstyle{every state}=[fill=black,draw=none,inner sep=0pt,minimum size=0.2cm]
\tikzstyle{every node}=[node distance=\hdist]
\node[state](a0)[fill=\vfour] at (0,0) {};
\tikzstyle{every node}=[node distance=\vdist]
\node[state](a1)[below of=a0]{};
\draw[\colw] (a0) -- (a1);
\tikzstyle{every node}=[node distance=\hdist]

%%%%%%%%%%%%%%Check
\tikzstyle{every state}=[fill=\colcheck,draw=none,inner sep=0pt,minimum size=0.2cm]
\tikzstyle{every node}=[node distance=\vdist]
\node[state](cb)[right of=a1]{};
\node(d1)[above of=cb]{};
\node(d2)[above of=d1]{};
\node(d3)[above of=d2]{};
\node[state](ct)[above of=d3]{};
\draw[\colcheck] (cb) -- (ct);
\node(cc1) at (0.0,3.5*\vdist) {$\textcolor{\colcheck}{ \chk_{1}}$ };

%%%%%%%%%%%%%%% answer
\tikzstyle{every node}=[node distance=\hdist]
\node[state](ab)[right of=cb]{};
\node[state](at)[right of=ct]{};
\draw[\colcheck] (ab) -- (at);
\node(cc2) at (0.7cm + 2*\hdist,3.5*\vdist) {$\textcolor{\colcheck}{ \answer_{1,{\bf c}}}$ };
%%%%%%%%%%%%%%%%Win
\begin{scope}[shift={($(5*\hdist,4*\vdist)$)}]
\totowin
\node(lab1) at (2*\hdist,-3*\vdist) {$\textcolor{blue}{\win\text{ or }\lose}$ };
\end{scope}

\end{scope}
\end{scope}
\end{scope}
\end{scope}

\begin{scope}[shift={($(l1)+(\hdist,0)$)}]
\oneround
\node(lab1) at (0,0.5) {$\increment_{T,0}$ };
\node(lab2) at (\hdist,-0.2) {$\increment_{T,1}$ };
\node(lab3) at (2 * \hdist,0.5 ) {$\increment_{T,2}$ };

\begin{scope}[shift={($(3*\hdist,0)$)}]
\oneround
\begin{scope}[shift={($(3*\hdist,0)$)}]
\tikzstyle{every state}=[fill=black,draw=none,inner sep=0pt,minimum size=0.2cm]
\tikzstyle{every node}=[node distance=\hdist]
\node[state](a0)[fill=\vfour] at (0,0) {};
\tikzstyle{every node}=[node distance=\vdist]
\node[state](a1)[below of=a0]{};
\draw[\colw] (a0) -- (a1);
\tikzstyle{every node}=[node distance=\hdist]
\node[state](b1)[right of=a1]{};
\tikzstyle{every node}=[node distance=\vdist]
\node[state](b2)[below of=b1]{};
\draw[\colw] (b1) -- (b2);

\end{scope}
\end{scope}
\end{scope}

\end{scope}
\end{tikzpicture}
    \caption{An uncontrollable check of processes $T_1$ and $B_1$ after $3$ rounds of the top pool and $4$ rounds of the bottom pool. Intuitively, the environment is asking to processes $T_1$ and $B_1$ ``what is the value of $f(3,4)$, where $f:[n]\times [m]\to C$ is your solution to the constrained coloring problem?''. The check is followed by the controllable action $\answer_{1,c}$ of the same two processes, parametrized by a color $c\in C$. This way the processes claim that $f(3,4)=c$. The answer is followed by either the \win\ or the \lose\ action, which terminates the game.}%
    \label{fig:check}
\end{figure}
\end{center}

For every index $\ell\in\{0,1,2\}$,
there is an uncontrollable action $\chk_\ell$ which synchronizes the two processes
$T_\ell$ and $B_\ell$.
A transition on  action $\chk_\ell$ is possible if and only if
the last action
of both processes
is either an increment or the action \terminate.
In particular, both processes should have played at least one increment.
After the check,
the only possible transitions rely on controllable actions $\answer_{\ell,c}$ synchronizing $T_\ell$ and $B_\ell$
and indexed by colors $c\in C$.
The answer is stored in the state space of the processes.

After that, there are only two possible outcomes: either the $\win$ or the $\lose$ action synchronizes all processes, and then no further transition is available.
There is no extra condition needed to execute a global transition on $\win$:
as soon as at least one pair of processes has performed
an answer the transition on \win\ is available.

\paragraph*{Six ways to lose.}

Considering the non-deterministic environment as adversarial,
it prefers transitions on $\lose$ rather than those on $\win$.
Losing can occur in six different ways,
five of them correspond to the five constraints in the definition of the \lcp.

\newcommand{\ri}{R}
A transition on $\lose$ is possible in case one of the following conditions holds:
\begin{enumerate}
\item[(a)]
if both processes $T_\ell$ and $B_\ell$ have played a single increment
and their answer is a color $c$
which is not initial (i.e. $c\not \in C_i$).
The condition ``has played a single increment''
can be stored in the state space of the processes,
and used to allow or not this transition; or
\item[(b)]
if both processes $T_\ell$ and $B_\ell$
have played $\terminate$ and their answer is a color $c$
which is not final (i.e. $c\not \in C_f$).
\end{enumerate}

\noindent
There are four other ways to lose,
which require that two checks occur in parallel,
for two different pairs of processes indexed by $\ell$ and $\ell'$, respectively.
This leads to two parallel answers corresponding to colors $c$ and $c'$, respectively.
The conditions for losing rely on the definition of the \emph{round index} of a play $u$ for a process $X_\ell$,
defined as $\lfloor(h-1)/2\rfloor$ where $h>0$ is the number of increments played by $X_\ell$ in $u$ when the check occurs.
The round index is denoted  $\ri_{X_\ell}(u)$.

\begin{figure}[ht]
    \begin{tikzpicture}
\begin{scope}[scale=1]

\newcommand{\eeee}{mycolor}
\newcommand{\eeeeee}{mycolor}
\newcommand{\eeeee}{1.5cm}
\newcommand{\vfour}{black}
\newcommand{\vfourr}{mycolor}
\newcommand{\colw}{black}
\newcommand{\colwin}{ForestGreen}
\newcommand{\colend}{orange}
\newcommand{\colcheck}{BurntOrange}
\newcommand{\collose}{red}
\newcommand{\spec}{\vfour}
\newcommand{\ff}[1]{{}}

\newcommand{\vdist}{.6cm}
\newcommand{\hdist}{.7cm}

\tikzstyle{every node}=[node distance=\vdist]
\node(lab1) at (1,-1) {$T_0$ };
\node(lab2) [below of=lab1] {$T_1$ };
\node(lab3) [below of=lab2] {$T_2$ };
\node(lab0) [below of=lab3] { };
\node(lab4) [below of=lab0] {$B_0$ };
\node(lab5) [below of=lab4] {$B_1$ };
\node(lab6) [below of=lab5] {$B_2$ };

\tikzstyle{every node}=[node distance=.2cm]
\node(l1) [right of=lab1] {};
\node(l2) [right of=lab2] {};
\node(l3) [right of=lab3] {};
\node(l4) [right of=lab4] {};
\node(l5) [right of=lab5] {};
\node(l6) [right of=lab6] {};

%%%%% the horizontal lines
\tikzstyle{every node}=[node distance=12cm]
\node(r1) [right of=l1] {};
\draw[gray] (l1) -- (r1);
\node(r2) [right of=l2] {};
\draw[gray] (l2) -- (r2);
\node(r3) [right of=l3] {};
\draw[gray] (l3) -- (r3);
\node(r4) [right of=l4] {};
\draw[gray] (l4) -- (r4);
\node(r5) [right of=l5] {};
\draw[gray] (l5) -- (r5);
\node(r6) [right of=l6] {};
\draw[gray] (l6) -- (r6);

\newcommand{\oneround}{
\tikzstyle{every state}=[fill=black,draw=none,inner sep=0pt,minimum size=0.2cm]
\tikzstyle{every node}=[node distance=\hdist]
\node[state](a0)[fill=\vfour] at (0,0) {};
\node(dummy)[right of=a0] {};
\node[state](c0)[right of=dummy] {};
\tikzstyle{every node}=[node distance=\vdist]
\node[state](a1)[below of=a0]{};
\draw[\colw] (a0) -- (a1);
\tikzstyle{every node}=[node distance=\hdist]
\node[state](b1)[right of=a1]{};
\tikzstyle{every node}=[node distance=\vdist]
\node[state](b2)[below of=b1]{};
\draw[\colw] (b1) -- (b2);
\tikzstyle{every node}=[node distance=\hdist]
\node[state](c2)[right of=b2]{};
\draw[\colw] (c0) -- (c2);
%%%%the separation ation
\tikzstyle{every node}=[node distance=0.5*\hdist]
\node(sep0)[right of=c0]{};
\node(sep1)[right of=c2]{};
\tikzstyle{every node}=[node distance=0.8*\vdist]
\node(sepa0)[above of=sep0]{};
\node(sepa1)[below of=sep1]{};
\draw[dashed, \colw] (sepa0) -- (sepa1);
}

\newcommand{\totoend}{
\tikzstyle{every state}=[fill=\colend,draw=none,inner sep=0pt,minimum size=0.2cm]
\tikzstyle{every node}=[node distance=\vdist]
\node[state](e0)at (0,0) {};
\node[state](e1)[below of=e0] {};
\node[state](e2)[below of=e1]{};
}

\newcommand{\totosync}[2]{
\tikzstyle{every state}=[fill=#1,draw=none,inner sep=0pt,minimum size=0.2cm]
\tikzstyle{every node}=[node distance=\vdist]
\node[state](e0)at (0,0) {};
\node[state](e1)[below of=e0] {};
\node[state](e2)[below of=e1]{};
\node(dummy)[below of=e2]{};
\node[state](eb0)[below of=dummy] {};
\node[state](eb1)[below of=eb0] {};
\node[state](eb2)[below of=eb1]{};
\draw[#1] (e0) -- (e1) -- (e2) -- (eb0) -- (eb1) -- (eb2) ;
\node(lab1) at (0.5,-3*\vdist) {$\textcolor{#1}{#2}$ };

}

\newcommand{\totowin}{\totosync{\colwin}{\win}}
\newcommand{\totolose}{\totosync{\collose}{\lose}}

\newcommand{\drawchk}[2] {
%%%%%%%%%%%%%%Check1
\tikzstyle{every state}=[fill=\colcheck,draw=none,inner sep=0pt,minimum size=0.2cm]
\tikzstyle{every node}=[node distance=\vdist]
\node[state](cb)at (0,0){};
\node(d1)[above of=cb]{};
\node(d2)[above of=d1]{};
\node(d3)[above of=d2]{};
\node[state](ct)[above of=d3]{};
\draw[\colcheck] (cb) -- (ct);
\node(cc1) at (-\hdist,4.5*\vdist) {$\textcolor{\colcheck}{ \chk_{#1}}$ };

%%%%%%%%%%%%%%% answer1
\tikzstyle{every node}=[node distance=\hdist]
\node[state](ab)[right of=cb]{};
\node[state](at)[right of=ct]{};
\draw[\colcheck] (ab) -- (at);
\node(cc2) at (2.5*\hdist,4.5*\vdist) {$\textcolor{\colcheck}{ \answer_{#1,{\bf #2}}}$ };
}

\begin{scope}[shift={($(l4)+(\hdist,0)$)}]
\oneround
\node(lab1) at (0,0.5) {$\increment_{B,0}$ };
\node(lab2) at (\hdist,-0.2) {$\increment_{B,1}$ };
\node(lab3) at (2 * \hdist,0.5 ) {$\increment_{B,2}$ };
\begin{scope}[shift={($(3*\hdist,0)$)}]
\oneround
\begin{scope}[shift={($(3*\hdist,0)$)}]
\oneround
\begin{scope}[shift={($(3*\hdist,0)$)}]
\tikzstyle{every state}=[fill=black,draw=none,inner sep=0pt,minimum size=0.2cm]
\tikzstyle{every node}=[node distance=\hdist]
\node[state](a0)[fill=\vfour] at (0,0) {};
\tikzstyle{every node}=[node distance=\vdist]
\node[state](a1)[below of=a0]{};
\draw[\colw] (a0) -- (a1);
\tikzstyle{every node}=[node distance=\hdist]

\begin{scope}[shift={($(a1)+(\hdist,0)$)}]
\drawchk{1}{c}
\end{scope}

\begin{scope}[shift={($(a0)+(1.3*\hdist,0)$)}]
\drawchk{0}{d}
\end{scope}

%%%%%%%%%%%%%%%%Lose
\begin{scope}[shift={($(5.5*\hdist,4*\vdist)$)}]
\totolose
\end{scope}

\end{scope}
\end{scope}
\end{scope}
\end{scope}

\begin{scope}[shift={($(l1)+(\hdist,0)$)}]
\oneround
\node(lab1) at (0,0.5) {$\increment_{T,0}$ };
\node(lab2) at (\hdist,-0.2) {$\increment_{T,1}$ };
\node(lab3) at (2 * \hdist,0.5 ) {$\increment_{T,2}$ };

\begin{scope}[shift={($(3*\hdist,0)$)}]
\oneround
\begin{scope}[shift={($(3*\hdist,0)$)}]
\tikzstyle{every state}=[fill=black,draw=none,inner sep=0pt,minimum size=0.2cm]
\tikzstyle{every node}=[node distance=\hdist]
\node[state](a0)[fill=\vfour] at (0,0) {};
\tikzstyle{every node}=[node distance=\vdist]
\node[state](a1)[below of=a0]{};
\draw[\colw] (a0) -- (a1);
\tikzstyle{every node}=[node distance=\hdist]
\node[state](b1)[right of=a1]{};
\tikzstyle{every node}=[node distance=\vdist]
\node[state](b2)[below of=b1]{};
\draw[\colw] (b1) -- (b2);

\end{scope}
\end{scope}
\end{scope}

\end{scope}
\end{tikzpicture}
    \caption{One of six ways to lose: two checks occur in parallel.
    In both pools, the two checks occur during the same round.
    The two pairs of processes give (in parallel as well) two different answers $c\neq d$.
    This allows the environment to trigger the \lose\ action. Intuitively, the players have cheated since they gave two different answers to the question ``what is the value of $f(3,4)$''?}%
    \label{fig:checkcheat}
\end{figure}
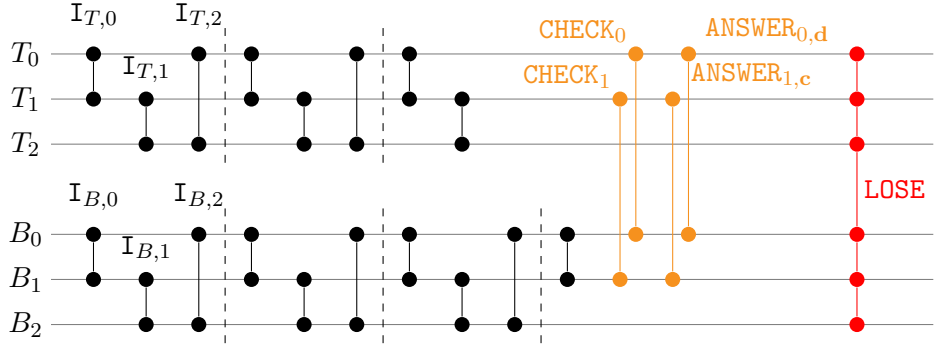

The conditions for \lose\ are based on a comparison of $c$ and $c'$ as well as the respective rounds of the processes.
For each pool $X\in\{T,B\}$ and every play $u$ there are two possibilities.
W.l.o.g., assume we have chosen $\ell$ and $\ell'$ so that $\ell' < \ell$.
Either both processes $X_\ell$ and $X_{\ell'}$ are in the same round
(meaning they have the same round index in $u$)
or process $X_\ell'$ is one round ahead of $X_{\ell}$
(meaning that $\ri_{X_\ell'}(u) = 1 + \ri_{X_\ell}(u)$).
No other case may occur because
every uninterrupted play in pool $X$ is a prefix of a word
in $(\increment_{X,0}\increment_{X,1}\increment_{X,2})^*\terminate$
and $\ell' < \ell$.

A transition on $\lose$ is possible in case one of the following conditions holds:
\begin{enumerate}
\item[(c)]
in both pools both processes are in the same round but the answers are different (i.e. $c\neq c'$); or
\item[(d)]
in both pools process $X_\ell'$ is one round ahead of $X_\ell$ and the pair of answers $(c,c')$ is a forbidden square (i.e. $(c,c')\in S)$; or
\item[(e)]
in the top pool process $T_\ell'$ is one round ahead of $T_\ell$ while in the bottom pool
both processes are in the same round and the pair of answers $(c,c')$ is a forbidden upper-triangle (i.e. $(c,c')\in UT)$; or
\item[(f)]
in the top pool both processes are in the same round while in the bottom pool
process $B_\ell'$ is one round ahead of $B_\ell$
and the pair of answers $(c,c')$ is a forbidden lower-triangle (i.e. $(c,c')\in LT)$.
\end{enumerate}

\noindent
Remark that the condition ``being in the same round''
can be easily implemented in the transition table.
For that,
it is enough that each process keeps track
in its state space
of the number of increments it has already played, modulo $4$.
From that counter, one can derive the parity of the round index
and compare the respective parities for the two processes $X_\ell$ and $X_{\ell'}$.

Condition (c) is illustrated on Figure~\ref{fig:checkcheat}
while condition (d) is illustrated on Figure~\ref{fig:checksquare}.

\begin{center}
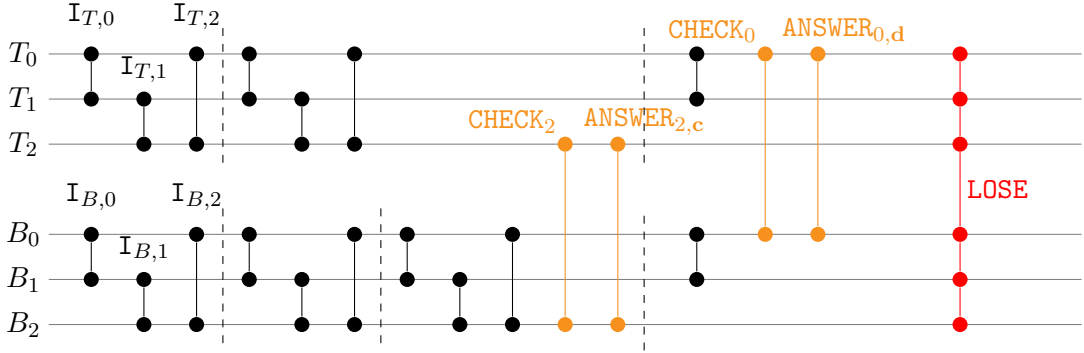
\begin{figure}[ht]
    \begin{tikzpicture}
\begin{scope}[scale=1]

\newcommand{\eeee}{mycolor}
\newcommand{\eeeeee}{mycolor}
\newcommand{\eeeee}{1.5cm}
\newcommand{\vfour}{black}
\newcommand{\vfourr}{mycolor}
\newcommand{\colw}{black}
\newcommand{\colwin}{ForestGreen}
\newcommand{\colend}{orange}
\newcommand{\colcheck}{BurntOrange}
\newcommand{\collose}{red}
\newcommand{\spec}{\vfour}
\newcommand{\ff}[1]{{}}

\newcommand{\vdist}{.6cm}
\newcommand{\hdist}{.7cm}

\tikzstyle{every node}=[node distance=\vdist]
\node(lab1) at (1,-1) {$T_0$ };
\node(lab2) [below of=lab1] {$T_1$ };
\node(lab3) [below of=lab2] {$T_2$ };
\node(lab0) [below of=lab3] { };
\node(lab4) [below of=lab0] {$B_0$ };
\node(lab5) [below of=lab4] {$B_1$ };
\node(lab6) [below of=lab5] {$B_2$ };

\tikzstyle{every node}=[node distance=.2cm]
\node(l1) [right of=lab1] {};
\node(l2) [right of=lab2] {};
\node(l3) [right of=lab3] {};
\node(l4) [right of=lab4] {};
\node(l5) [right of=lab5] {};
\node(l6) [right of=lab6] {};

%%%%% the horizontal lines
\tikzstyle{every node}=[node distance=14cm]
\node(r1) [right of=l1] {};
\draw[gray] (l1) -- (r1);
\node(r2) [right of=l2] {};
\draw[gray] (l2) -- (r2);
\node(r3) [right of=l3] {};
\draw[gray] (l3) -- (r3);
\node(r4) [right of=l4] {};
\draw[gray] (l4) -- (r4);
\node(r5) [right of=l5] {};
\draw[gray] (l5) -- (r5);
\node(r6) [right of=l6] {};
\draw[gray] (l6) -- (r6);

\newcommand{\oneround}{
\tikzstyle{every state}=[fill=black,draw=none,inner sep=0pt,minimum size=0.2cm]
\tikzstyle{every node}=[node distance=\hdist]
\node[state](a0)[fill=\vfour] at (0,0) {};
\node(dummy)[right of=a0] {};
\node[state](c0)[right of=dummy] {};
\tikzstyle{every node}=[node distance=\vdist]
\node[state](a1)[below of=a0]{};
\draw[\colw] (a0) -- (a1);
\tikzstyle{every node}=[node distance=\hdist]
\node[state](b1)[right of=a1]{};
\tikzstyle{every node}=[node distance=\vdist]
\node[state](b2)[below of=b1]{};
\draw[\colw] (b1) -- (b2);
\tikzstyle{every node}=[node distance=\hdist]
\node[state](c2)[right of=b2]{};
\draw[\colw] (c0) -- (c2);
%%%%the separation ation
}

\newcommand{\totosep}{
\tikzstyle{every node}=[node distance=0.5*\hdist]
\node(sep0)[right of=c0]{};
\node(sep1)[right of=c2]{};
\tikzstyle{every node}=[node distance=0.8*\vdist]
\node(sepa0)[above of=sep0]{};
\node(sepa1)[below of=sep1]{};
\draw[dashed, \colw] (sepa0) -- (sepa1);
}

\newcommand{\totoend}{
\tikzstyle{every state}=[fill=\colend,draw=none,inner sep=0pt,minimum size=0.2cm]
\tikzstyle{every node}=[node distance=\vdist]
\node[state](e0)at (0,0) {};
\node[state](e1)[below of=e0] {};
\node[state](e2)[below of=e1]{};
}

\newcommand{\totosync}[2]{
\tikzstyle{every state}=[fill=#1,draw=none,inner sep=0pt,minimum size=0.2cm]
\tikzstyle{every node}=[node distance=\vdist]
\node[state](e0)at (0,0) {};
\node[state](e1)[below of=e0] {};
\node[state](e2)[below of=e1]{};
\node(dummy)[below of=e2]{};
\node[state](eb0)[below of=dummy] {};
\node[state](eb1)[below of=eb0] {};
\node[state](eb2)[below of=eb1]{};
\draw[#1] (e0) -- (e1) -- (e2) -- (eb0) -- (eb1) -- (eb2) ;
\node(lab1) at (0.5,-3*\vdist) {$\textcolor{#1}{#2}$ };

}

\newcommand{\totowin}{\totosync{\colwin}{\win}}
\newcommand{\totolose}{\totosync{\collose}{\lose}}

\newcommand{\drawchk}[2] {
%%%%%%%%%%%%%%Check1
\tikzstyle{every state}=[fill=\colcheck,draw=none,inner sep=0pt,minimum size=0.2cm]
\tikzstyle{every node}=[node distance=\vdist]
\node[state](cb)at (0,0){};
\node(d1)[above of=cb]{};
\node(d2)[above of=d1]{};
\node(d3)[above of=d2]{};
\node[state](ct)[above of=d3]{};
\draw[\colcheck] (cb) -- (ct);
\node(cc1) at (-\hdist,4.5*\vdist) {$\textcolor{\colcheck}{ \chk_{#1}}$ };

%%%%%%%%%%%%%%% answer1
\tikzstyle{every node}=[node distance=\hdist]
\node[state](ab)[right of=cb]{};
\node[state](at)[right of=ct]{};
\draw[\colcheck] (ab) -- (at);
\node(cc2) at (1.5*\hdist,4.5*\vdist) {$\textcolor{\colcheck}{ \answer_{#1,{\bf #2}}}$ };
}

\begin{scope}[shift={($(l4)+(\hdist,0)$)}]
\oneround
\totosep
\node(lab1) at (0,0.5) {$\increment_{B,0}$ };
\node(lab2) at (\hdist,-0.2) {$\increment_{B,1}$ };
\node(lab3) at (2 * \hdist,0.5 ) {$\increment_{B,2}$ };
\begin{scope}[shift={($(3*\hdist,0)$)}]
\oneround
\totosep
\begin{scope}[shift={($(3*\hdist,0)$)}]
\oneround

%%%%%%%%%%%%%%%%%%%% check 2
\begin{scope}[shift={($(a0)+(3*\hdist,-2*\vdist)$)}]
\drawchk{2}{c}
\tikzstyle{every node}=[node distance=0.5*\hdist]
\node(sep1)[right of=ab]{};
\tikzstyle{every node}=[node distance=2*\vdist]
\node(sep0)[above of=sep1]{};
\tikzstyle{every node}=[node distance=0.8*\vdist]
\node(sepa0)[below of=sep1]{};
\node(sepa1)[above of=sep0]{};
\draw[dashed, \colw] (sepa0) -- (sepa1);
\end{scope}

%%%last increment of B
\begin{scope}[shift={($(5*\hdist,0)$)}]
\tikzstyle{every state}=[fill=black,draw=none,inner sep=0pt,minimum size=0.2cm]
\tikzstyle{every node}=[node distance=\hdist]
\node[state](a0)[fill=\vfour] at (0.5*\hdist ,0) {};
\tikzstyle{every node}=[node distance=\vdist]
\node[state](a1)[below of=a0]{};
\draw[\colw] (a0) -- (a1);

%%%%%%%%%%%%%%%%Lose
\begin{scope}[shift={($(5.5*\hdist,4*\vdist)$)}]
\totolose
\end{scope}

\end{scope}
\end{scope}
\end{scope}
\end{scope}

\begin{scope}[shift={($(l1)+(\hdist,0)$)}]
\oneround
\totosep
\node(lab1) at (0,0.5) {$\increment_{T,0}$ };
\node(lab2) at (\hdist,-0.2) {$\increment_{T,1}$ };
\node(lab3) at (2 * \hdist,0.5 ) {$\increment_{T,2}$ };

\begin{scope}[shift={($(3*\hdist,0)$)}]
\oneround

%%%%%%%%%%%%%%%%%%%% check 0

\begin{scope}[shift={($(8.5*\hdist,0)$)}]

%%%last increment of T
\tikzstyle{every state}=[fill=black,draw=none,inner sep=0pt,minimum size=0.2cm]
\tikzstyle{every node}=[node distance=\hdist]
\node[state](a0)[fill=\vfour] at (0,0) {};
\tikzstyle{every node}=[node distance=\vdist]
\node[state](a1)[below of=a0]{};
\draw[\colw] (a0) -- (a1);
\tikzstyle{every node}=[node distance=\hdist]

\tikzstyle{every node}=[node distance=\hdist]
\node(sep1)[left of=a0]{};
\tikzstyle{every node}=[node distance=2*\vdist]
\node(sep0)[below of=sep1]{};
\tikzstyle{every node}=[node distance=0.8*\vdist]
\node(sepa0)[above of=sep1]{};
\node(sepa1)[below of=sep0]{};
\draw[dashed, \colw] (sepa0) -- (sepa1);

\begin{scope}[shift={($(a1)+(1.3*\hdist,-3*\vdist)$)}]
\drawchk{0}{d}
\end{scope}

\end{scope}
\end{scope}
\end{scope}

\end{scope}
\end{tikzpicture}
    \caption{Another way to lose: two checks occur in parallel. $\chk_2$ occurs in round $2$ of pool $T$ and round $3$ of pool $B$ thus $T_2,B_2$ are assumed to answer with $d=f(2,3)$. In parallel, $\chk_0$ occurs in round $3$ of pool $T$ and round $4$ of pool $B$. Hence $T_0,B_0$ are assumed to answer with $d=f(3,4)$. We assume here that the pairs of answers $(c,d)$ is a forbidden square, which  allows the environment to trigger the \lose\ action.}%
\label{fig:checksquare}
\end{figure}
\end{center}

\subsection{Proof of Lemma~\ref{lem:equiv}}

Lemma~\ref{lem:equiv} is an equivalence.
We start with the converse implication.
Assume  that there is a finite bipartite coloring $f:[n]\times [m]\to C$ satisfying the constraints $(C_i,C_f,S,UT,LT)$.
A winning strategy for processes of the top (resp.\ bottom) pool, as long as they are not interrupted by a check, consists in playing $n$ rounds (resp. $m$ rounds) of increments, followed by the \terminate\ action.
If the environment triggers a check on $T_\ell$
and $B_\ell$ after some play $u$, the processes answer $f(x,y)$, where $x$ (resp. $y$) is the
round index of $T_\ell$ (resp.\ of $B_\ell$) in $u$. This information is available to both processes since they share the same causal past right after the check.
Since $f$ satisfies the constraint, the environment cannot trigger a transition on \lose.
Condition (a) cannot occur because all the answers when both processes are in round $0$ are equal to $f(0,0)\in C_i$.
Condition (b) cannot occur because all the answers when both processes are in round $(n,m)$ are equal to $f(n,m)\in C_f$.
Condition (c) cannot occur because the answers only depends on the rounds, independently of the identities of the processes
and the number of increments they have performed in the current round.
Condition (d), (e) and (f) cannot occur because no pattern induced by $f$ is forbidden.

\medskip

We now prove the direct implication of Lemma~\ref{lem:equiv}.
Assume that processes have a winning strategy $\sigma$.
The first step is to define a finite bipartite $C$-coloring $f$.
Let $n$ (resp. $m$) be round index of $T_1$ (resp.\ of $B_1$) in the maximal
uninterrupted play consistent with $\sigma$ (which exists since $\sigma$ is winning).
For every $(x,y)\in [n]\times [m]$ denote
\begin{align*}
&u_{T,x}=(\increment_{T,0}\increment_{T,1}\increment_{T,2})^{x}\increment_{T,0}\increment_{T,1}\\
&u_{B,y}=(\increment_{B,0}\increment_{B,1}\increment_{B,2})^{y}\increment_{B,0}\increment_{B,1}\\
&u_{x,y} \text{ the parallel product of } u_{T,x} \text{ and } u_{B,y}
\end{align*}
and let $f(x,y)$ be the answer of
$T_1,B_1$
after a check on the play $u_{x,y}$ i.e.
\[
f(x,y) = c \text { such that } (u_{x,y}\chk_1\answer_{1,c}) \text{ is a $\sigma$-play}  \enspace.
\]
This is well defined: since $\sigma$ is winning it creates no deadlock in a non-final state, thus
at least one answer of $T_1,B_1$ is allowed by $\sigma$  after $u_{x,y}\chk_1$. In the sequel, we assume that exactly one answer is allowed by $\sigma$. This is w.l.o.g.:\ if there is a winning strategy which allows several answers, we can restrict it to allow a single answer after every check, and it will still be winning.

We show that $f$ satisfies the five constraints $(C_i,C_f,S,UT,LT)$.

\paragraph{Constraint $C_i$ on initial colors.}
Since $\sigma$ is winning, no transition on action \lose\ is available
after $u_{0,0}\chk_1\answer_{1,f(0,0)}$. Thus according to
condition (a) above,
$f(0,0)$ is an initial color.

\paragraph{Constraint $C_f$ on final colors.}
We consider the checks
\begin{align*}
&p_1 = u_{n,m}\chk_1\\
&p_2 =u_{n,m}\increment_{T,2}\increment_{B,2}\chk_2\\
&p_3 =u_{n,m}\terminate_{T,1}\terminate_{B,1}\chk_1\enspace.
\end{align*}

By definition of $f$, the answer to $p_1$ is $f(n,m)$.
Denote $c_2$ and $c_3$
the answers to $p_2$ and $p_3$, respectively.
Remark that $p_1$ and $p_2$ can occur in parallel,
because processes $T_1,B_1$ do not play in
$\increment_{T,2}\increment_{B,2}\chk_2$.
For similar reasons, also $p_2$ and $p_3$ can occur in parallel.
Moreover, in all these three plays, processes $T_1$ and $T_2$ are in the same round,
and processes $B_1$ and $B_2$ are also in the same round
(remember that \terminate\ actions are not accounted for when computing the round index).
Since $\sigma$ is winning,
condition (c) will neither be satisfied by
the pair of parallel checks $p_1,p_2$ nor by the pair $p_2,p_3$,
thus $f(n,m)=c_2=c_3$.
Finally, after the check $p_3$, both processes $T_1,B_1$
are in the state \sss\ and they answer $f(n,m)$.
According to condition (b) above,
$f(n,m)$ is a final color.

\paragraph{Constraint $S$ on forbidden squares.}
Let $x,y\in [n-1]\times[m-1]$.
We consider four possible checks
\begin{align*}
& p_1 = u_{x,y}\chk_1\\
&p_2 = u_{x,y}\increment_{T,2}\increment_{B,2}\chk_2\\
&p_0' = u_{x,y}\increment_{T,2}\increment_{B,2}\increment_{T,0}\increment_{B,0}\chk_0\\
& p_1' = u_{x+1,y+1}\chk_1 \enspace.
\end{align*}
By definition of $f$, the answers to $p_1$ and $p'_1$ are
$f(x,y)$ and $f(x+1,y+1)$, respectively. Denote $c$ the answer to $p_2$
and $d$ the answer to $p_0'$.
The four checks are  illustrated on Figure~\ref{fig:squaretrans}.

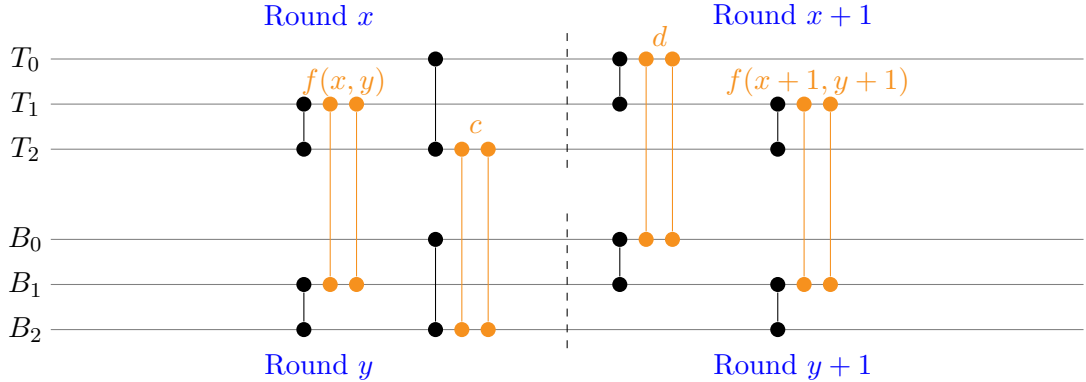
\begin{figure}[ht]
    \begin{tikzpicture}
\begin{scope}[scale=1]

\newcommand{\eeee}{mycolor}
\newcommand{\eeeeee}{mycolor}
\newcommand{\eeeee}{1.5cm}
\newcommand{\vfour}{black}
\newcommand{\vfourr}{mycolor}
\newcommand{\colw}{black}
\newcommand{\colwin}{ForestGreen}
\newcommand{\colend}{orange}
\newcommand{\colcheck}{BurntOrange}
\newcommand{\collose}{red}
\newcommand{\spec}{\vfour}
\newcommand{\ff}[1]{{}}

\newcommand{\vdist}{.6cm}
\newcommand{\hdist}{.7cm}

\tikzstyle{every node}=[node distance=\vdist]
\node(lab1) at (1,-1) {$T_0$ };
\node(lab2) [below of=lab1] {$T_1$ };
\node(lab3) [below of=lab2] {$T_2$ };
\node(lab0) [below of=lab3] { };
\node(lab4) [below of=lab0] {$B_0$ };
\node(lab5) [below of=lab4] {$B_1$ };
\node(lab6) [below of=lab5] {$B_2$ };

\tikzstyle{every node}=[node distance=.2cm]
\node(l1) [right of=lab1] {};
\node(l2) [right of=lab2] {};
\node(l3) [right of=lab3] {};
\node(l4) [right of=lab4] {};
\node(l5) [right of=lab5] {};
\node(l6) [right of=lab6] {};

%%%%% the horizontal lines
\tikzstyle{every node}=[node distance=14cm]
\node(r1) [right of=l1] {};
\draw[gray] (l1) -- (r1);
\node(r2) [right of=l2] {};
\draw[gray] (l2) -- (r2);
\node(r3) [right of=l3] {};
\draw[gray] (l3) -- (r3);
\node(r4) [right of=l4] {};
\draw[gray] (l4) -- (r4);
\node(r5) [right of=l5] {};
\draw[gray] (l5) -- (r5);
\node(r6) [right of=l6] {};
\draw[gray] (l6) -- (r6);

\newcommand{\oneround}{
\tikzstyle{every state}=[fill=black,draw=none,inner sep=0pt,minimum size=0.2cm]
\node[state](b1)at (0,-\vdist){};
\tikzstyle{every node}=[node distance=\vdist]
\node[state](b2)[below of=b1]{};
\draw[\colw] (b1) -- (b2);
\tikzstyle{every node}=[node distance=2.5*\hdist]
\node[state](c2)[right of=b2]{};
\tikzstyle{every node}=[node distance=2*\vdist]
\node[state](c0)[above of=c2] {};
\draw[\colw] (c0) -- (c2);

%%%%%%%%%%%%Sep
\tikzstyle{every node}=[node distance=2.5*\hdist]
\node(sep0)[right of=c0]{};
\node(sep1)[right of=c2]{};
\tikzstyle{every node}=[node distance=0.8*\vdist]
\node(sepa0)[above of=sep0]{};
\node(sepa1)[below of=sep1]{};
\draw[dashed, \colw] (sepa0) -- (sepa1);

%%%%%%%%%%%% x+1
\begin{scope}[shift={($(6*\hdist,0)$)}]
\tikzstyle{every state}=[fill=black,draw=none,inner sep=0pt,minimum size=0.2cm]
\tikzstyle{every node}=[node distance=2*\hdist]
\node[state](a0)[fill=\vfour] at (0,0) {};
\tikzstyle{every node}=[node distance=\vdist]
\node[state](a1)[below of=a0]{};
\draw[\colw] (a0) -- (a1);
\tikzstyle{every node}=[node distance=3*\hdist]
\node[state](b1)[right of=a1]{};
\tikzstyle{every node}=[node distance=\vdist]
\node[state](b2)[below of=b1]{};
\draw[\colw] (b1) -- (b2);
\end{scope}
%%%%the separation ation
}

\newcommand{\totosep}{
\tikzstyle{every node}=[node distance=0.5*\hdist]
\node(sep0)[right of=c0]{};
\node(sep1)[right of=c2]{};
\tikzstyle{every node}=[node distance=0.8*\vdist]
\node(sepa0)[above of=sep0]{};
\node(sepa1)[below of=sep1]{};
\draw[dashed, \colw] (sepa0) -- (sepa1);
}

%%%%%%%%%%%%%% Draw Check answer
\newcommand{\drawchk}[2] {

\tikzstyle{every state}=[fill=\colcheck,draw=none,inner sep=0pt,minimum size=0.2cm]
\tikzstyle{every node}=[node distance=\vdist]
\node[state](cb)at (0,0){};
\node(d1)[above of=cb]{};
\node(d2)[above of=d1]{};
\node(d3)[above of=d2]{};
\node[state](ct)[above of=d3]{};
\draw[\colcheck] (cb) -- (ct);
%\node(cc1) at (-\hdist,2*\vdist) {$\textcolor{\colcheck}{ \chk_{#1}}$ };

\tikzstyle{every node}=[node distance=0.5*\hdist]
\node[state](ab)[right of=cb]{};
\node[state](at)[right of=ct]{};
\draw[\colcheck] (ab) -- (at);
\node(cc2) at (0.3*\vdist,4.5*\vdist) {$\textcolor{\colcheck}{ #2}$ };
}

\begin{scope}[shift={($(l4)+(5*\hdist,0)$)}]
\oneround
\begin{scope}[shift={($(0.5*\hdist,-\vdist)$)}]
\drawchk{1}{f(x,y)}
\end{scope}
\begin{scope}[shift={($(3*\hdist,-2*\vdist)$)}]
\drawchk{2}{c}
\end{scope}
\begin{scope}[shift={($(6.5*\hdist,0)$)}]
\drawchk{0}{d}
\end{scope}
\begin{scope}[shift={($(9.5*\hdist,-\vdist)$)}]
\drawchk{1}{f(x+1,y+1)}
\end{scope}
\end{scope}

\begin{scope}[shift={($(l1)+(5*\hdist,0)$)}]
\oneround
\end{scope}

\node(lab1) at (7*\hdist,-0.7*\vdist) {\textcolor{blue}{Round $x$} };
\node(lab2) at (7*\hdist,-8.5*\vdist) {\textcolor{blue}{Round $y$} };
\node(lab1) at (16*\hdist,-0.7*\vdist) {\textcolor{blue}{Round $x+1$} };
\node(lab2) at (16*\hdist,-8.5*\vdist) {\textcolor{blue}{Round $y+1$} };

\end{scope}
\end{tikzpicture}
    \caption{The four checks used to enforce the constraint on forbidden squares,
    and their answers.
    These four checks cannot happen altogether in the same play, but any pair of two successive checks can happen in parallel, in the same play.
    }%
\label{fig:squaretrans}
\end{figure}

The two checks $p_1$ and $p_2$ can occur in parallel during the same round,
thus according to condition (c), $f(x,y)=c$.
The two checks $p_2$ and $p_0'$ can occur in parallel and the check $p_0'$
is one round ahead of the check $p_2$, with respect to both pools.
Since $\sigma$ is winning, condition (d) cannot happen
hence $(c,d)$ is not a forbidden square.
The two checks $p_0'$ and $p_1'$ may occur in parallel during the same round,
thus according to condition (c), $d=f(x+1,y+1)$.
Finally $(f(x,y),f(x+1,y+1))$ is not a forbidden square.

\paragraph{Constraint $UT$ on forbidden upper-triangles.}
Let $x,y\in [n-1]\times[m]$.
We consider the plays
\begin{align*}
&p_1 = u_{x,y}\chk_1\\
&p_2 = u_{x,y}\increment_{T,2}\increment_{B,2}\chk_2\\
&p_0' = u_{x,y}\increment_{T,2}\increment_{B,2}\increment_{T,0}\chk_0\\
&p_1' = u_{x+1,y}\chk_1\enspace.
\end{align*}
By definition of $f$, the answers to $p_1$ and $p'_1$ are
$f(x,y)$ and $f(x+1,y)$, respectively. Denote $c$ the answer to $p_2$
and $d$ the answer to $p_0'$.
The four checks are  illustrated on Figure~\ref{fig:squareup}.

\begin{center}
\begin{figure}[ht]
    \begin{tikzpicture}
\begin{scope}[scale=1]

\newcommand{\eeee}{mycolor}
\newcommand{\eeeeee}{mycolor}
\newcommand{\eeeee}{1.5cm}
\newcommand{\vfour}{black}
\newcommand{\vfourr}{mycolor}
\newcommand{\colw}{black}
\newcommand{\colwin}{ForestGreen}
\newcommand{\colend}{orange}
\newcommand{\colcheck}{BurntOrange}
\newcommand{\collose}{red}
\newcommand{\spec}{\vfour}
\newcommand{\ff}[1]{{}}

\newcommand{\vdist}{.6cm}
\newcommand{\hdist}{.7cm}

\tikzstyle{every node}=[node distance=\vdist]
\node(lab1) at (1,-1) {$T_0$ };
\node(lab2) [below of=lab1] {$T_1$ };
\node(lab3) [below of=lab2] {$T_2$ };
\node(lab0) [below of=lab3] { };
\node(lab4) [below of=lab0] {$B_0$ };
\node(lab5) [below of=lab4] {$B_1$ };
\node(lab6) [below of=lab5] {$B_2$ };

\tikzstyle{every node}=[node distance=.2cm]
\node(l1) [right of=lab1] {};
\node(l2) [right of=lab2] {};
\node(l3) [right of=lab3] {};
\node(l4) [right of=lab4] {};
\node(l5) [right of=lab5] {};
\node(l6) [right of=lab6] {};

%%%%% the horizontal lines
\tikzstyle{every node}=[node distance=14cm]
\node(r1) [right of=l1] {};
\draw[gray] (l1) -- (r1);
\node(r2) [right of=l2] {};
\draw[gray] (l2) -- (r2);
\node(r3) [right of=l3] {};
\draw[gray] (l3) -- (r3);
\node(r4) [right of=l4] {};
\draw[gray] (l4) -- (r4);
\node(r5) [right of=l5] {};
\draw[gray] (l5) -- (r5);
\node(r6) [right of=l6] {};
\draw[gray] (l6) -- (r6);

\newcommand{\onesmallround}{
\tikzstyle{every state}=[fill=black,draw=none,inner sep=0pt,minimum size=0.2cm]
\node[state](b1)at (0,-\vdist){};
\tikzstyle{every node}=[node distance=\vdist]
\node[state](b2)[below of=b1]{};
\draw[\colw] (b1) -- (b2);
\tikzstyle{every node}=[node distance=2.5*\hdist]
\node[state](c2)[right of=b2]{};
\tikzstyle{every node}=[node distance=2*\vdist]
\node[state](c0)[above of=c2] {};
\draw[\colw] (c0) -- (c2);
}
\newcommand{\oneround}{
\onesmallround
%%%%%%%%%%%%Sep
\tikzstyle{every node}=[node distance=2.5*\hdist]
\node(sep0)[right of=c0]{};
\node(sep1)[right of=c2]{};
\tikzstyle{every node}=[node distance=0.8*\vdist]
\node(sepa0)[above of=sep0]{};
\node(sepa1)[below of=sep1]{};
\draw[dashed, \colw] (sepa0) -- (sepa1);

%%%%%%%%%%%% x+1
\begin{scope}[shift={($(6*\hdist,0)$)}]
\tikzstyle{every state}=[fill=black,draw=none,inner sep=0pt,minimum size=0.2cm]
\tikzstyle{every node}=[node distance=2*\hdist]
\node[state](a0)[fill=\vfour] at (0,0) {};
\tikzstyle{every node}=[node distance=\vdist]
\node[state](a1)[below of=a0]{};
\draw[\colw] (a0) -- (a1);
\tikzstyle{every node}=[node distance=3*\hdist]
\node[state](b1)[right of=a1]{};
\tikzstyle{every node}=[node distance=\vdist]
\node[state](b2)[below of=b1]{};
\draw[\colw] (b1) -- (b2);
\end{scope}
%%%%the separation ation
}

\newcommand{\totosep}{
\tikzstyle{every node}=[node distance=0.5*\hdist]
\node(sep0)[right of=c0]{};
\node(sep1)[right of=c2]{};
\tikzstyle{every node}=[node distance=0.8*\vdist]
\node(sepa0)[above of=sep0]{};
\node(sepa1)[below of=sep1]{};
\draw[dashed, \colw] (sepa0) -- (sepa1);
}

%%%%%%%%%%%%%% Draw Check answer
\newcommand{\drawchk}[2] {

\tikzstyle{every state}=[fill=\colcheck,draw=none,inner sep=0pt,minimum size=0.2cm]
\tikzstyle{every node}=[node distance=\vdist]
\node[state](cb)at (0,0){};
\node(d1)[above of=cb]{};
\node(d2)[above of=d1]{};
\node(d3)[above of=d2]{};
\node[state](ct)[above of=d3]{};
\draw[\colcheck] (cb) -- (ct);
%\node(cc1) at (-\hdist,2*\vdist) {$\textcolor{\colcheck}{ \chk_{#1}}$ };

\tikzstyle{every node}=[node distance=0.5*\hdist]
\node[state](ab)[right of=cb]{};
\node[state](at)[right of=ct]{};
\draw[\colcheck] (ab) -- (at);
\node(cc2) at (0.3*\vdist,4.5*\vdist) {$\textcolor{\colcheck}{ #2}$ };
}

\begin{scope}[shift={($(l4)+(5*\hdist,0)$)}]
\onesmallround
\begin{scope}[shift={($(0.5*\hdist,-\vdist)$)}]
\drawchk{1}{f(x,y)}
\end{scope}
\begin{scope}[shift={($(3*\hdist,-2*\vdist)$)}]
\drawchk{2}{c}
\end{scope}
\begin{scope}[shift={($(6.5*\hdist,0)$)}]
\drawchk{0}{d}
\end{scope}
\begin{scope}[shift={($(9.5*\hdist,-\vdist)$)}]
\drawchk{1}{f(x+1,y)}
\end{scope}
\end{scope}

\begin{scope}[shift={($(l1)+(5*\hdist,0)$)}]
\oneround
\end{scope}

\node(lab1) at (7*\hdist,-0.7*\vdist) {\textcolor{blue}{Round $x$} };
\node(lab2) at (7*\hdist,-8.5*\vdist) {\textcolor{blue}{Round $y$} };
\node(lab2) at (16*\hdist,-0.7*\vdist) {\textcolor{blue}{Round $x+1$} };

\end{scope}
\end{tikzpicture}
    \caption{The four checks used to enforce the constraint on forbidden upper-triangles, and their answers.
    These four checks cannot happen altogether in the same play, but any pair of two successive checks can happen in parallel, in the same play.
    }%
\label{fig:squareup}
\end{figure}
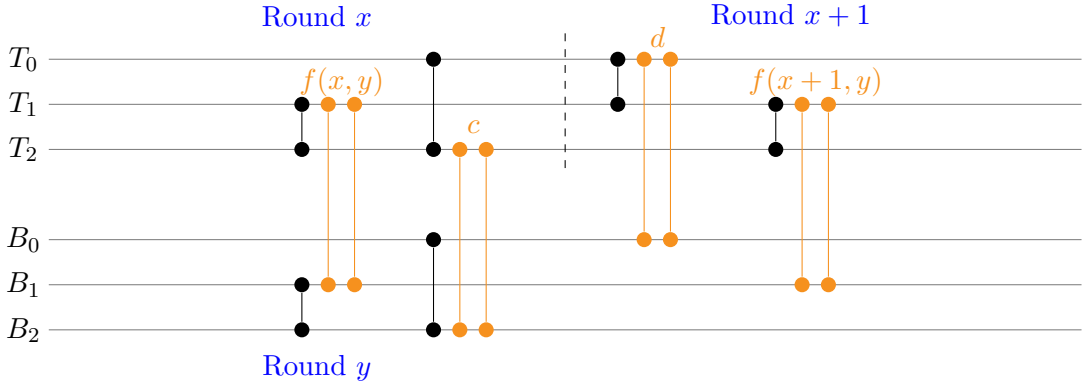
\end{center}
The two checks $p_1$ and $p_2$ can occur in parallel during the same round
thus according to condition (c), $f(x,y)=c$.
The two checks $p_2$ and $p_0'$ can occur in parallel,
with process $T_0$ one round ahead of $T_2$ and processes
$B_0$ and $B_2$ in the same round
thus according to condition (e), $(c,d)$ is \emph{not} a forbidden upper-triangle.
The two checks $p_0'$ and $p_1'$ can occur in parallel during the same round
thus according to condition (c), $d=f(x+1,y)$.
Finally, $(f(x,y),f(x+1,y))$ is not a forbidden upper-triangle.

\paragraph{Constraint $LT$ on forbidden lower-triangles.}
This case is symmetric with the $UT$ case above.

\medskip

Finally, $f$ is a finite bipartite coloring satisfying all six constraints.
This terminates the proof of the direct implication of Lemma~\ref{lem:equiv}.\qed%

\section*{Conclusion}
We have proved that the \dsp\  is undecidable for six processes,
thus closing an open problem which was first raised in~\cite{gastin}
and has afterwards been shown decidable in several special cases~\cite{madhu,acyclic,DBLP:conf/fsttcs/Gimbert17,beutner2019translating}.
The proof relies on a reduction from the \PCP, via an intermediary decision problem called the \lcp. A direct reduction from the \PCP\ is of course also possible, but in our opinion it gives a less clear picture of the reasons for the undecidability.

The construction shows undecidability of the \dsp\ when the winning condition
is local termination
(every process eventually end up in a final state).
Another natural condition in this setting
is the \emph{deadlock-freeness condition}, either global (the play is infinite)
or local (every process plays infinitely often).
The construction can be easily adapted to show undecidability
for deadlock-freeness as well,
using an encoding of the infinite version of the \PCP\@.

An interesting open problem is the decidability of the \dsp\ when
the automaton is~\emph{grid-free},
i.e.\ when there exists a bound $B$
such that whenever two \emph{parallel} runs are followed by a synchronization,
one of the two runs has length $\leq B$ (see~\cite[Definition 1, Corollary 5]{thiagarajan2014rabin} for a precise definition).
The construction in this paper is clearly not grid-free:
the two pools can run independently in parallel for an arbitrary long time
before being synchronized by the environment.
The existence of a winning strategy can be stated
as the satisfiability
of a Monadic Second Order
Logic (MSOL) formula
but this does not imply decidability,
 since MSOL is undecidable for grid-free systems~\cite{chalopin20191}.
 However, there are classes of systems with decidable control
 but undecidable MSOL satisfiability,
for example the \dsp\ is decidable for systems with four processes~\cite{DBLP:conf/fsttcs/Gimbert17} while MSOL satisfiability is undecidable
with only two processes~\cite{thiagarajan2014rabin}.
This might be the case for grid-free systems as well.

The decidability of the \dsp\ for five processes remains an open question.
Is it decidable like in the case of four processes?
The decidability proof in~\cite{DBLP:conf/fsttcs/Gimbert17}
relies on a small-model property of the winning strategies.
When two processes do synchronize, they acquire simultaneously perfect knowledge about their respective states, hence an upper-bound on the number of transitions needed by two processes to win. When there are only four processes, two processes
who synchronize can base their strategic decision on their current states and the boundedly many possible plays of the two other processes, hence a bound on the number of transitions needed by the four players to win.
This  property does not seem easy to get with five processes.

Is the~\dsp\ undecidable for five processes, then?
 At first sight there is no easy way to adapt the present construction
with two pools of three processes
to show the undecidability for five processes.
If the bottom pool is reduced to two processes only,
it seems hard to enforce the lower-triangle constraints of the \lcp.
When the two sole processes of the bottom pool
synchronize they acquire perfect knowledge of
the state of their pool, and it is not possible anymore to constrain their future
behaviour with respect to their past behaviour: what was done is gone.

\bibliographystyle{alphaurl}

\bibliography{trombones}

\newpage

\appendix

%\section{Proof of Theorem~\ref{thm:reduction}}
%\label{app:reduc}

%\section{Under the hood}
%\label{app:uth}
%\input{underthehood}

\end{document}